\documentclass[letter]{article}
\usepackage{fullpage}
\usepackage{graphicx}
\usepackage{amsmath}
\usepackage{amssymb}
\usepackage{amsthm}
\usepackage{xspace}
\usepackage{figlatex}
\usepackage{subfigure}
\usepackage{rotating}

\newcommand{\loadmatrix}{\ensuremath{A}\xspace}
\newcommand{\sizeX}{\ensuremath{n_1}\xspace}
\newcommand{\sizeY}{\ensuremath{n_2}\xspace}
\newcommand{\sizeOneD}{\ensuremath{n}\xspace}
\newcommand{\nbproc}{\ensuremath{m}\xspace}
\newcommand{\maxload}{\ensuremath{L_{max}}\xspace}
\newcommand{\avgload}{\ensuremath{L_{avg}}\xspace}
\newcommand{\prefixsum}{\ensuremath{\Gamma}\xspace}
\newcommand{\load}{\ensuremath{L}\xspace}
\newcommand{\pxq}{\ensuremath{P\!\times\!Q}\xspace}
\def\naive{na\"{\i}ve\xspace}

\def\naively{na\"{\i}vely\xspace}

\newcommand{\ceil}[1]{\lceil#1\rceil}
\newcommand{\Ceil}[1]{\left\lceil#1\right\rceil}

\newcommand{\Floor}[1]{\left\lfloor#1\right\rfloor}

\newtheorem{lemma}{Lemma}
\newtheorem{theorem}{Theorem}

\begin{document}

\title{Load-Balancing Spatially Located Computations using Rectangular
  Partitions\footnote{This work was supported in parts by the U.S. DOE
    SciDAC Institute Grant DE-FC02-06ER2775; by the U.S. National
    Science Foundation under Grants CNS-0643969, OCI-0904809 and
    OCI-0904802.}}

\author{Erik Saule$^1$, Erdeniz \"{O}. Ba\c{s}$^{1,2}$ and \"Umit
  V. \c{C}ataly\"urek$^{1,3}$\\
  email: {\textit \{esaule,erdeniz,umit\}@bmi.osu.edu}\\
  $^1$Department of Biomedical Informatics. The Ohio State
  University\\
  $^2$Department of Computer Science and Engineering. The Ohio
  State University\\
  $^3$Department of Electrical and Computer Engineering. The
  Ohio State University}

\maketitle

\begin{abstract}
 Distributing spatially located heterogeneous workloads is an
 important problem in parallel scientific computing. We investigate
 the problem of partitioning such workloads (represented as a matrix
 of non-negative integers) into rectangles, such that the load of the
 most loaded rectangle (processor) is minimized. Since finding the
 optimal arbitrary rectangle-based partition is an NP-hard problem, we
 investigate particular classes of solutions: rectilinear, jagged and
 hierarchical. We present a new class of solutions called \nbproc-way
 jagged partitions, propose new optimal algorithms for \nbproc-way
 jagged partitions and hierarchical partitions, propose new heuristic
 algorithms, and provide worst case performance analyses for some
 existing and new heuristics. Moreover, the algorithms are tested in
 simulation on a wide set of instances. Results show that two of the
 algorithms we introduce lead to a much better load balance than the
 state-of-the-art algorithms. We also show how to design a two-phase
 algorithm that reaches different time/quality tradeoff.

\end{abstract}

\section{Introduction}
\label{sec:intro}

To achieve good efficiency when using a parallel platform, one must
distribute the computations and the required data to the processors of the parallel
machine. If the computation tasks are independent, their parallel
processing falls in the category of pleasantly parallel tasks. Even in
such cases, when computation time of the tasks are not equal,
obtaining the optimal load balance to achieve optimum execution time becomes computationally hard and
heuristic solutions are used~\cite{Handbookchap3}. Furthermore, most
of the time some dependencies exist between the tasks and some data
must be shared or exchanged frequently, making the problem even more
complicated.

A large class of application see its computations take place in a
geometrical space of typically two or three dimensions. Different
types of applications fall into that class. Particle-in-cell
simulation~\cite{Plimpton03,Karimabadi06} is an implementation of the
classical mean-field approximation of the many-body problem in
physics. Typically, thousands to millions of particles are located in
cells, which are a discretization of a field. The application
iteratively updates the value of the field in a cell, based on the
state of the particles it contains and the value of the neighboring
cells, and then the state of the particles, based on its own state and
the state of the cell it belongs to. Direct volume
rendering~\cite{Kutluca00} is an application that use rendering
algorithm similar to raycasting in a scene of semi-transparent objects
without reflection. For each pixel of the screen, a ray orthogonal to
the screen is cast from the pixel and color information will be
accumulated over the different objects the ray encounters. Each pixel
therefore requires an amount of computation linear to the number of
objects crossed by the ray and neighboring pixels are likely to cross
the same objects. Partial Differential Equation can be computed using
mesh-based computation. For instance~\cite{Horak05} solves heat
equation on a surface by building a regular mesh out of it. The state of
each node of the mesh is iteratively updated depending on the state of
neighboring nodes. For load balancing purpose, \cite{Nicol94}~maps the
mesh to a discretized two dimensional space. An other application can
be found in 3D engines where the state of the world is iteratively
updated and where the updates on each object depends on neighboring
objects (for instance, for collision purpose)~\cite{Abdelkhalek04}. Linear algebra
operations can potentially also benefit from such
techniques~\cite{Vastenhouw05,Pinar97,Ujaldon96}.

In this work, our goal is to balance the load of such applications.  In the
literature, load balancing techniques can be broadly divided into two
categories: geometric and connectivity-based. Geometric methods (such
as~\cite{Berger87,DBLP:journals/tpds/PilkingtonB96}) leverages the
fact that computations which are close by in the space are more likely
to share data than computations that are far in the space, by dividing
the load using geometric properties of the workload.  Methods from
that class often rely on a recursive decomposition of the domain such
as octrees~\cite{Flaherty97adaptivelocal} or they rely on space
filling curves and surfaces~\cite{Aftosmis04,Aluru97paralleldomain}.
Connectivity-based methods usually model the load balancing problem
through a graph or an hypergraph weighted with computation volumes on
the nodes and communication volumes on the edges or hyper edges (see
for instance~\cite{schloegel00sc,Catalyurek09JPDC}).
Connectivity-based techniques lead to good partitions but are usually
computationally expensive and require to build an accurate graph (or
hypergraph) model of the computation.  They are particularly
well-suited when the interactions between tasks are irregular. Graphs
are useful when modeling interactions that are exactly between two
tasks, and hypergraph are useful when modeling more complex
interactions that could involve more than two
tasks~\cite{Catalyurek99,Hendrickson00}.

When the interactions are regular (structured) one can use methods
that takes the structure into account. For example, when coordinate
information for tasks are available, one can use geometric methods
which leads to ``fast'' and effective partitioning techniques.  In
geometric partitioning, one prefers to partition the problem into
connex and compact parts so as to minimize communication volumes.
Rectangles (and rectangular volumes) are the most preferred shape
because they implicitly minimize communication, do not restrict the
set of possible allocations drastically, are easily expressed and
allow to quickly find which rectangle a coordinate belongs to using
simple data structures. Hence, in this work, we will only focus
partitioning into rectangles.

In more concrete terms, this paper addresses the problem of
partitioning a two-dimensional load matrix composed of non-negative
numbers into a given number of rectangles (processors) so as to minimize the load
of the most loaded rectangle; the most loaded rectangle is the one
whose sum of the element it contains is maximal. The problem is
formulated so that each element of the array represents a task and
each rectangle represents a processor. Computing the optimal solution
for this problem has been shown to be
NP-Hard~\cite{Grigni96}. Therefore, we focus on algorithms with low
polynomial complexity that lead to good solutions.

The approach we are pursuing in this work is to consider different
classes of rectangular partitioning. Simpler structures are expected
to yield bad load balance but to be computed quickly while more
complex structures are expected to give good load balance but lead to
higher computation time. For each class, we look for optimal
algorithms and heuristics. Several algorithms to deal with this
particular problem which have been proposed in the literature are 
described and analyzed. One original class of solution is
proposed and original algorithms are presented and analyzed.

The theoretical analysis of the algorithms is accompanied by an
extensive experimentation evaluation of the algorithms to decide which
one should be used in practice. The experimentation is composed of
various randomly generated datasets and two datasets extracted from two
applications, one following the particle-in-cell paradigm and one
following the mesh-based computation paradigm.

The contributions of this work are as follows:
\begin{itemize}
\item A classical \pxq-way jagged heuristic is theoretically analyzed
  by bounding the load imbalance it generates in the worst case.
\item We propose a new class of solutions, namely, \nbproc-way jagged
  partitions, for which we propose a fast heuristic as well as an
  exact polynomial dynamic programming formulation. This heuristic is
  also theoretically analyzed and shown to perform better than the
  \pxq-way jagged heuristic.
\item For an existing class of solutions, namely, hierarchical
  bipartitions, we propose both an optimal polynomial dynamic
  programming algorithm as well as a new heuristic.
\item The presented and proposed algorithms are practically assessed
  in simulations performed on synthetic load matrices and on real load
  matrices extracted from both a particle-in-cell simulator and a
  geometric mesh. Simulations show that two of the proposed heuristics
  outperform all the tested existing algorithms.
\item Algorithmic engineering techniques are used to create hybrid
  partitioning scheme that provides slower algorithms but with
  higher quality.
\end{itemize}

This work extends~\cite{Saule11-IPDPS} by providing the following main
contributions: tighter bounds in the theoretical guarantee of an
\nbproc-way jagged partitioning heuristic, new heuristics for
\nbproc-way jagged partitioning, experimental results of proposed
algorithms with detailed charts, and hybrid algorithms.

Several previous work tackles a similar problem but they usually
presents only algorithms from one class with no experimental
validation or a very simple one. These works are referenced in the
text when describing the algorithm they introduce. Kutluca et
al.~\cite{Kutluca00} is the closest related work. They are
tackling the parallelization of a Direct Volume 
Rendering application whose load balancing is done using a very
similar model. They survey rectangle based partition but also more
general partition generated from hypergraph modeling and space
filling curves. The experimental validation they propose is based on the
actual runtime of the Direct Volume Rendering application.

Similar classes of solutions are used in the problem of partitioning an
equally loaded tasks onto heterogeneous processors (see
\cite{LD09chap3} for a survey). This is a very different problem which often
assumes the task space is continuous (therefore infinitely
divisible). Since the load balance is trivial to optimize in such a
context, most work in this area focus on optimizing communication
patterns.

The rest of the paper is organized as follows. Section~\ref{sec:model}
presents the model and notations used. The different classes of
partitions are described in Section~\ref{sec:algo}. This section also
presents known and new polynomial time algorithms either optimal or
heuristic. The algorithms are evaluated in Section~\ref{sec:expe} on
synthetic dataset as well as on dataset extracted from two real
simulation codes. Section~\ref{sec:hybrid} presents a two-phase
technique, namely hybrid algorithms, to generate
partitions. Conclusive remarks are given in Section~\ref{sec:ccl}.

\section{Model and Preliminaries}
\label{sec:model}

\subsection{Problem Definition}

Let \loadmatrix be a two dimensional array of  $\sizeX \times
\sizeY$ non-negative integers representing the spatially located load.
This load matrix needs to be distributed on \nbproc processors. Each
element of the array must be allocated to exactly one processor. The
load of a processor is the sum of the elements of the array it has
been allocated. The cost of a solution is the load of the most loaded
processor. The problem is to find a solution that minimizes the cost.

In this paper we are only interested in rectangular allocations, and
we will use 'rectangle' and 'processor' interchangeably. That
is to say, a solution is a set $R$ of \nbproc rectangles $r_i = (x_1,
x_2, y_1, y_2)$ which form a partition of the elements of the array.
Two properties have to be ensured for a solution to be valid:
$\bigcap_{r \in R} = \emptyset$ and $\bigcup_{r \in R} = A$. The first
one can be checked by verifying that no rectangle collides with
another one, it can be done using line to line tests and inclusion
test. The second one can be checked by verifying that all the
rectangles are inside $A$ and that the sum of their area is equal to
the area of $A$. This testing method runs in $O(\nbproc^2)$. The load
of a processor is $\load (r_i) = \sum_{x_1 \leq x \leq x_2} \sum_{y_1
  \leq y \leq y_2} \loadmatrix [x][y]$. The load of the most loaded
processor in solution $R$ is $\maxload = \max_{r_i} \load (r_i)$. We
will denote by $\maxload^*$ the minimal cost
achievable. Notice that $\maxload^* \geq \frac{ \sum_{x,y} \loadmatrix
  [x][y] }{\nbproc}$ and $\maxload^* \geq \max_{x,y} \loadmatrix
[x][y]$ are lower bounds of the optimal maximum load. In term of
distributed computing, it is important to remark that this model is
only concerned by computation times and not by communication
times.

Algorithms that tackle this problem rarely consider the load
of a single element of the matrix. Instead, they  usually consider
 the load of a rectangle. Therefore, we assume that matrix
$\loadmatrix$ is given as a 2D prefix sum array $\prefixsum$ so that
$\prefixsum [x][y] = \sum_{x' \leq x, y'\leq y} \loadmatrix
[x'][y']$. That way, the load of a rectangle $r = (x_1, x_2, y_1,
y_2)$ can be computed in $O(1)$ (instead of $O((x_2-x_1)
(y_2-y_1))$), as 
$\load (r) = \prefixsum [x_2][y_2]
- \prefixsum [x_1-1][y_2] - \prefixsum [x_2][y_1-1] + \prefixsum
[x_1-1][y_1-1]$. 

An algorithm $H$ is said to be a $\rho$-approximation algorithm, if for
all instances of the problem, it returns a solution which maximum load
is no more than $\rho$ times the optimal maximum load, i.e., $\maxload(H) \leq
\rho \maxload^*$. In simulations, the metric used for qualifying the
solution is the load imbalance which is computed as
$\frac{\maxload}{\avgload} -1$ where $\avgload = \frac{\sum_{x,y}
  \loadmatrix [x][y]}{\nbproc}$. A solution which is perfectly
balanced achieves a load imbalance of 0. Notice that the optimal
solution for the maximum load might not be perfectly balanced and
usually has a strictly positive load imbalance. The ratio of most
approximation algorithm are proved using \avgload as the only lower
bound on the optimal maximum load. Therefore, it usually means that a
$\rho$-approximation algorithm leads to a solution whose load
imbalance is less than $\rho-1$.

\subsection{The One Dimensional Variant}

Solving the 2D partitioning problem is obviously harder than solving
the 1D partitioning problem. Most of the algorithms for the 2D
partitioning problems are inspired by 1D partitioning algorithms.  An
extensive theoretical and experimental comparison of those 1D algorithms
has been given in~\cite{Pinar04}. In~\cite{Pinar04}, the fastest
optimal 1D partitioning algorithm is {\tt NicolPlus}; it is an
algorithmically engineered modification of \cite{Nicol94}, which uses
a subroutine proposed in~\cite{Han92}. A slower optimal algorithm
using dynamic programming was proposed in~\cite{Manne95}. Different
heuristics have also been
developed~\cite{Miguet97,Pinar04}. Frederickson~\cite{Frederickson1990}
proposed an $O(\sizeOneD)$ optimal algorithm which is only arguably
better than $O((\nbproc \log \frac{ \sizeOneD}{ \nbproc})^2)$ obtained
by {\tt NicolPlus}. Moreover, Frederickson's algorithm requires
complicated data structures which are difficult to implement and are
likely to run slowly in practice. Therefore, in the remainder of
the paper {\tt NicolPlus} is the algorithm used for solving one
dimensional partitioning problems.

In the one dimensional case, the problem is to partition the array
\loadmatrix composed of \sizeOneD positive integers into \nbproc
intervals. 

{\tt DirectCut (DC)} (called "Heuristic 1" in~\cite{Miguet97}) is the
fastest reasonable heuristic. It greedily
allocates to each processor the smallest interval $I = \{0, \dots
,i\}$ which load is more than $\frac{\sum_i \loadmatrix[i]}{\nbproc}$. This can be
done in $O(\nbproc \log \frac{ \sizeOneD }{ \nbproc })$ using binary
search on the prefix sum array and the slicing technique
of~\cite{Han92}. By construction, {\tt DC} is a 2-approximation
algorithm but more precisely, $\maxload(DC) \leq \frac{\sum_i
  \loadmatrix [i]}{\nbproc} + \max_{i} \loadmatrix [i]$. This result
is particularly important since it provides an upper bound on the
optimal maximum load: $\maxload^* \leq \frac{\sum_i \loadmatrix
  [i]}{\nbproc} + \max_{i} \loadmatrix [i]$.

A widely known heuristic is {\tt Recursive Bisection (RB)} which
recursively splits the array into two parts of similar load and
allocates half the processors to each part. This algorithm leads to a
solution such that $\maxload(RB) \leq \frac{\sum_i \loadmatrix
  [i]}{\nbproc} + \max_{i} \loadmatrix [i]$ and therefore is a
2-approximation algorithm~\cite{Pinar04}. It has a runtime complexity
of $O(\nbproc \log \sizeOneD )$.

The optimal solution can be computed using dynamic
programming~\cite{Manne95}. The formulation comes from the property of
the problem that one interval must finish at index \sizeOneD. Then,
the maximum load is either given by this interval or by the maximum
load of the previous intervals. In other words, $\maxload^*(\sizeOneD,\nbproc)
= \min_{0 \leq k < \sizeOneD} \max \{ \maxload^* (k, \nbproc-1), \load (\{k+1,
\dots, \sizeOneD\})\}$. A detailed analysis shows that this formulation
leads to an algorithm of complexity $O(\nbproc (\sizeOneD - \nbproc)
)$.

The optimal algorithm in~\cite{Nicol94} relies on the parametric
search algorithm proposed in~\cite{Han92}. A function called
\textit{Probe} is given a targeted maximum load and either returns a
partition that reaches this maximum load or declares it
unreachable. The algorithm greedily allocates to each processor the
tasks and stops when the load of the processor will exceed the
targeted value. The last task allocated to a processor can be found
in $O(\log \sizeOneD)$ using a binary search on the prefix sum array,
leading to an algorithm of complexity $O(\nbproc \log \sizeOneD
)$. \cite{Han92} remarked that there are \nbproc binary searches which
look for increasing values in the array. Therefore, by slicing the
array in \nbproc parts, one binary search can be performed in $O(\log
\frac{ \sizeOneD }{\nbproc})$. It remains to decide in which part to
search for. Since there are \nbproc parts and the searched values are
increasing, it can be done in an amortized $O(1)$. This leads to a
\textit{Probe} function of complexity $O(\nbproc \log \frac{ \sizeOneD
}{\nbproc})$.

The algorithm proposed by~\cite{Nicol94} exploits the property that if
the maximum load is given by the first interval then its load is given
by the smallest interval so that $Probe(\load (\{0, \dots, i\}))$ is
true. Otherwise, the largest interval so that $Probe(\load (\{0,
\dots, i\}))$ is false can safely be allocated to the first
interval. Such an interval can be efficiently found using a binary
search, and the array slicing technique of~\cite{Han92} can be used to
reach a complexity of $O((\nbproc \log \frac{ \sizeOneD}{ \nbproc})^2
)$. Recent work~\cite{Pinar04} showed that clever
bounding techniques can be applied to reduce the range of the various
binary searches inside \textit{Probe} and inside the main function
leading to a runtime improvement of several orders of magnitude.

\section{Algorithms}
\label{sec:algo}

This section describes algorithms that can be used to solve the 2D
partitioning problem. These algorithms focus on generating a partition
with a given structure. Samples of the considered structures are presented in
Figure~\ref{fig:partitions}. Each structure is a
generalization of the previous one.

\begin{figure}
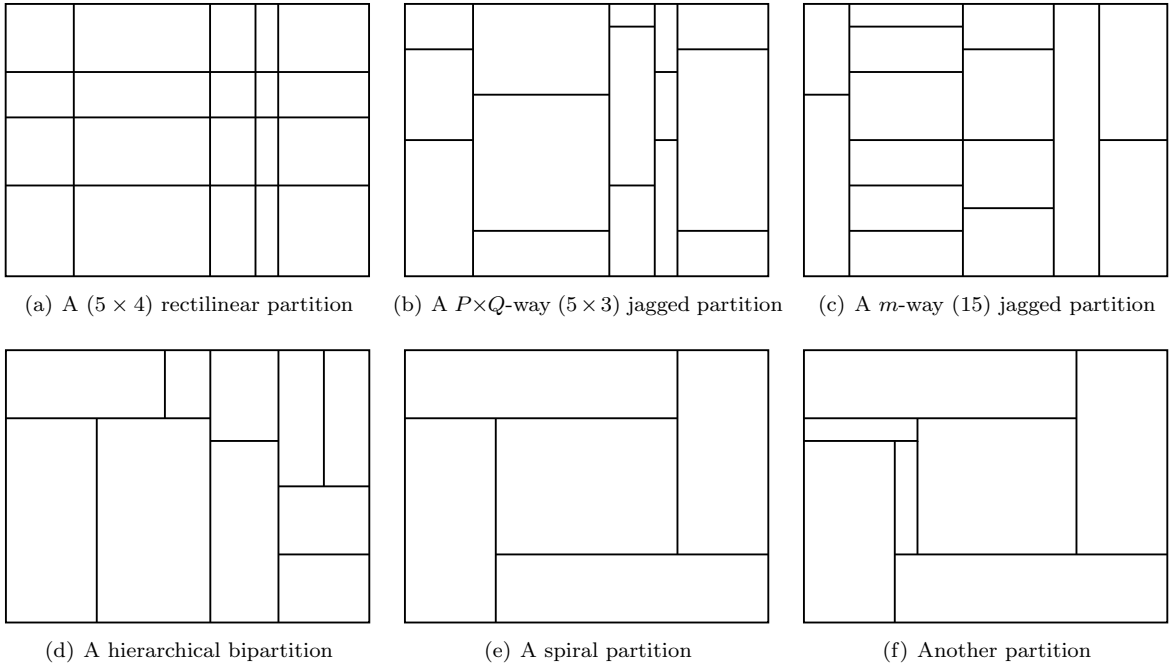

  \centering
  \subfigure[A ($5 \times 4$) rectilinear partition]{
    \includegraphics[width=.3\linewidth]{fig/rect_part.fig} \label{fig:rect_part}}
  \subfigure[A \pxq-way ($5 \times 3$) jagged partition]{
    \includegraphics[width=.3\linewidth]{fig/jagged_part_pxq.fig} \label{fig:jagged_part_pxq}}
  \subfigure[A \nbproc-way ($15$) jagged partition]{
    \includegraphics[width=.3\linewidth]{fig/jagged_part_mway.fig} \label{fig:jagged_part_mway}}
  \subfigure[A hierarchical bipartition]{
    \includegraphics[width=.3\linewidth]{fig/hierar_part.fig} \label{fig:hierar_part}}
  \subfigure[A spiral partition]{
    \includegraphics[width=.3\linewidth]{fig/spiral_part.fig} \label{fig:rec_part1}}
  \subfigure[Another partition]{
    \includegraphics[width=.3\linewidth]{fig/spiral_ext_part.fig} \label{fig:rec_part2}}
  \caption{Different structures of partitions.}
  \label{fig:partitions}
\end{figure}

Table~\ref{tab:all_algo} summarizes the different algorithms discussed in
this paper. Their worst-case complexity and theoretical guarantees are
given.

\begin{sidewaystable}
    \begin{tabular}{|l|c|c|l|}
      \hline Name & Parameters & Worst Case Complexity & Theoretical Guarantee\\\hline\hline
      {\tt RECT-UNIFORM} & P,Q & $O(PQ)$ & - \\\hline
      {\tt RECT-NICOL}\cite{Nicol94} & P,Q & $O(\sizeX \sizeY (Q(P\log{\frac{\sizeX}{P}})^2 + P(Q\log{\frac{\sizeY}{Q}})^2))$ & Better than {\tt RECT-UNIFORM}.\\\hline
      {\tt JAG-PQ-HEUR}\cite{Ujaldon96,Pinar97} & P,Q & $O((P \log \frac{\sizeX}{P})^2 + P (Q \log \frac{\sizeY}{Q})^2 )$ & [$\star$]$LI \leq (1+\Delta \frac{P}{\sizeX})(1+\Delta \frac{Q}{\sizeY}) -1$\\\hline
      {\tt JAG-PQ-OPT}\cite{Manne96,Pinar97} & P,Q & $O((P Q \log
\frac{\sizeX}{P} \log \frac{\sizeY}{Q})^2 )$ & Optimal for \pxq-way
jagged partitioning. \\&&or $O(\sizeX \log \sizeX (P + ( Q \log \frac{\sizeY}{Q})^2))$ &Better than all of the above.\\\hline
      {\tt JAG-M-HEUR}[$\star$] & P & $O((P \log \frac{\sizeX}{P})^2 + (\nbproc \log \frac{\sizeY}{\nbproc})^2 )$ & $LI \leq (\frac{\nbproc}{\nbproc-P} + \frac{\nbproc \Delta}{P \sizeY} + \frac{\Delta^2 \nbproc}{\sizeX\sizeY}) - 1$\\\hline
      {\tt JAG-M-PROBE}[$\star$] & P dividers & $O((\nbproc \log \frac{\sizeY P}{\nbproc})^2)$ & Optimal for the parameters.\\\hline
      {\tt JAG-M-ALLOC}[$\star$] & P, $Q_s$ & $O((P \log \frac{\sizeX}{P} \max_S Q_S \log \frac{\sizeY}{Q_S})^2)$ & Optimal for the parameters.\\\hline
      {\tt JAG-M-HEUR-PROBE}[$\star$] & P & $O((P \log \frac{\sizeX}{P})^2 +(\nbproc \log \frac{\sizeY P}{\nbproc})^2)$ & Better than {\tt JAG-M-HEUR}.\\\hline
      {\tt JAG-M-OPT}[$\star$] & - & $O(\sizeX^2 \nbproc^3 (\log \frac{ \sizeY}{ \nbproc })^2)$ & Optimal for \nbproc-way jagged partitioning.\\&&& Better than all of the above.\\\hline
      {\tt HIER-RB}\cite{Berger87} & - & $O(\nbproc \log \max (\sizeX, \sizeY))$ & -\\\hline
      {\tt HIER-RELAXED}[$\star$] & - & $O(\nbproc^2 \log\max (\sizeX,\sizeY) )$ & -\\\hline
      {\tt HIER-OPT}[$\star$] & - & $O(\sizeX^2 \sizeY^2 \nbproc^2 \log \max (\sizeX,\sizeY) )$ & Optimal for hierarchical bisection.\\&&&Better than all of the above.\\\hline
    \end{tabular}
 \caption{Summary of the presented algorithms. Algorithms and results introduced in this paper are denoted
   by a [$\star$]. $LI$ stands for Load Imbalance. $\Delta =
   \frac{\max_{i,j} \loadmatrix [i][j]}{\min_{i,j} \loadmatrix
     [i][j]}$, if $\forall i,j, \loadmatrix [i][j]>0$.}
 \label{tab:all_algo}
\end{sidewaystable}

\subsection{Rectilinear Partitions}

Rectilinear partitions (also called General Block Distribution
in~\cite{Aspvall01,Manne96}) organize the space according to a \pxq
grid as shown in Figure~\ref{fig:rect_part}. This type of partitions
is often used to optimize communication and indexing and has been
integrated in the High Performance Fortran standard~\cite{HPFF97}. It
is the kind of partition constructed by the MPI function {\tt
  MPI\_Cart}. This function is often implemented using the {\tt
  RECT-UNIFORM} algorithm which divides the first dimension and the
second dimension into $P$ and $Q$ intervals with size
$\frac{\sizeX}{P}$ and$\frac{\sizeY}{Q}$ respectively.  Notice that
{\tt RECT-UNIFORM} returns a \naive partition that balances the area
and not the load.

~\cite{Grigni96} implies that computing the optimal rectilinear partition is an NP-Hard
problem. \cite{Aspvall01}~points out that the
NP-completeness proof in~\cite{Grigni96} also implies that there is no
$(2-\epsilon)$-approximation algorithm unless P=NP. We can also remark
that the proof is valid for given values of $P$ and $Q$, but the complexity of
the problem is unclear if the only constraint is that $P Q \leq
\nbproc$. Notice that, the load matrix is often assumed to be a
square.

\cite{Nicol94}~(and~\cite{Manne96} independently) proposed an
iterative refinement heuristic algorithm that we call {\tt RECT-NICOL}
in the remaining of this paper. Provided the partition in one
dimension, called the fixed dimension, {\tt RECT-NICOL} computes the
optimal partition in the other dimension using an optimal one
dimension partitioning algorithm. The one dimension partitioning
problem is built by setting the load of an interval of the problem as
the maximum of the load of the interval inside each stripe of the
fixed dimension.  At each iteration, the partition of one dimension is
refined. The algorithm runs until the last 2 iterations return the
same partitions. Each iteration runs in
$O(Q(P\log{\frac{\sizeX}{P}})^2)$ or $O(P(Q\log{\frac{\sizeY}{Q}})^2)$
depending on the refined dimension. According to the analysis
in~\cite{Nicol94} the number of iterations is $O(\sizeX \sizeY)$ in
the worst case; however, in practice the convergence is faster (about
3-10 iterations for a 514x514 matrix up to 10,000
processors). \cite{Aspvall01}~shows it is a
$\theta(\sqrt{\nbproc})$-approximation when $P=Q=\sqrt{\nbproc}$.

The first constant approximation algorithm for rectilinear partitions
has been proposed by \cite{Khanna97} but neither the constant nor the
actual complexity is given. \cite{Aspvall01}~claims it is a
$120$-approximation that runs in $O(\sizeX \sizeY)$.

\cite{Aspvall01}~presents two different modifications of {\tt RECT-NICOL} which are
both a $\theta(\sqrt p)$-approximation algorithm for the rectilinear
partitioning problem of a $\sizeX \times \sizeX$ matrix in $p \times
p$ blocks which therefore is a $\theta(\nbproc^{1/4})$-approximation
algorithm. They run in a constant number of iterations (2 and 3) and
have a complexity of $O(\nbproc^{1.5}(\log{n})^2)$ and $O(n(\sqrt{\nbproc}
\log{n})^2)$. \cite{Aspvall01}~claims that despite the approximation ratio
is not constant, it is better in practice than the algorithm proposed
in \cite{Khanna97}.

\cite{Gaur02} provides a 2-approximation algorithm for the rectangle
stabbing problems which translates into a 4-approximation algorithm
for the rectilinear partitioning problem. This method is of high
complexity $O(\log(\sum_{i,j} \loadmatrix[i][j])\sizeX^{10}
\sizeY^{10})$ and heavily relies on linear programming to derive the
result.

\cite{Muthukrishnan05} considers resource augmentation and proposes a
$2$-approximation algorithm with slightly more processors than
allowed. It can be tuned to obtain a $(4+\epsilon)$-approximation
algorithm which runs in $O((n_1 + n_2 + P Q)P \log (n_1 n_2) )$.

\subsection{Jagged Partitions}

Jagged partitions (also called Semi Generalized Block Distribution
in~\cite{Manne96}) distinguish between the main dimension and the
auxiliary dimension. The main dimension will be split in $P$
intervals. Each rectangle of the solution must have its main dimension
matching one of these intervals. The auxiliary dimension of each
rectangle is arbitrary. Examples of jagged partitions are depicted in
Figures~\ref{fig:jagged_part_pxq} and~\ref{fig:jagged_part_mway}. The
layout of jagged partitions also allows to easily locate which
rectangle contains a given element~\cite{Ujaldon96}.

Without loss of generality, all the formulas in this section assume
that the main dimension is the first dimension.

\subsubsection{\pxq-way Jagged Partitions}

Traditionally, jagged partition algorithms are used to generate what
we call \pxq-way jagged partitions in which each interval of
the main dimension will be partitioned in $Q$ rectangles. Such a
partition is presented in Figure~\ref{fig:jagged_part_pxq}.

An intuitive heuristic to generate \pxq-way jagged
partitions, we call {\tt JAG-PQ-HEUR}, is to use a 1D partitioning
algorithm to partition the main dimension and then partition each
interval independently. First, we project the array on the main
dimension by summing all the elements along the auxiliary dimension.
An optimal 1D partitioning algorithm generates the intervals
of the main dimension. Then, for each interval, the elements are
projected on the auxiliary dimension by summing the elements along the
main dimension. An optimal 1D partitioning algorithm is used to
partition each interval. This heuristic have been proposed several
times before, for instance in~\cite{Ujaldon96,Pinar97}.

The algorithm runs in $O((P \log \frac{\sizeX}{P})^2 + P (Q \log
\frac{\sizeY}{Q})^2 )$. Prefix sum arrays avoid redundant projections: the load of interval
$(i,j)$ in the main dimension can be simply computed as $\load(i, j , 1, \sizeY)$.

We now provide an original analysis of the performance of this
heuristic under the hypothesis that all the elements of the load
matrix are strictly positive. First, we provide a refinement on the
upper bound of the optimal maximum load in the 1D partitioning problem
by refining the performance bound of {\tt DC} (and {\tt
  RB}) under this hypothesis.

\begin{lemma}
  \label{lem:DC-approx}
  If there is no zero in the array, applying {\tt DirectCut} on a one dimensional array
  $\loadmatrix$ using \nbproc processors leads to a maximum load having
  the following property: $\maxload(DC) \leq \frac{\sum \loadmatrix
    [i]}{\nbproc} (1+ \Delta \frac{\nbproc}{\sizeOneD})$ where $\Delta
  = \frac{\max_i \loadmatrix [i]}{\min_i \loadmatrix [i]}$.
\end{lemma}

\begin{proof}
  The proof is a simple rewriting of the performance bound of {\tt
    DirectCut}: $\maxload (DC)  \leq  \frac{\sum_i \loadmatrix
    [i]}{\nbproc} + \max_{i} \loadmatrix [i] \leq \frac{\sum_i \loadmatrix
    [i]}{\nbproc} (1+ \Delta \frac{\nbproc}{\sizeOneD})$.
\end{proof}

{\tt JAG-PQ-HEUR} is composed of two calls to an optimal one
dimensional algorithm. One can use the performance guarantee of {\tt
  DC} to bound the load imbalance at both steps. This is formally
expressed in the following theorem.

\begin{theorem}
  \label{th:pq-heur-approx}
  If there is no zero in the array, {\tt JAG-PQ-HEUR} is a $(1+\Delta
  \frac{P}{\sizeX})(1+\Delta \frac{Q}{\sizeY})$-approximation
  algorithm where $\Delta = \frac{\max_{i,j} \loadmatrix [i][j]}{\min_{i,j}
    \loadmatrix [i][j]}$, $P < \sizeX$, $Q < \sizeY$.
\end{theorem}

\begin{proof}
  Let us first give a bound on the load of the most loaded interval
  along the main dimension, i.e., the imbalance after the cut in the
  first dimension. Let $C$ denote the array of the projection of $A$
  among one dimension: $C[i] = \sum_j \loadmatrix[i][j]$. We have:
  $\maxload^*(C) \leq \frac{\sum_i C[i]}{P} (1+ \Delta
  \frac{P}{\sizeX})$. Noticing that $\sum_i C[i] = \sum_{i,j}
  A[i][j]$, we obtain: $\maxload^*(C) \leq \frac{\sum_{i,j}
    \loadmatrix [i][j]}{P} (1+ \Delta \frac{P}{\sizeX})$

  Let $S$ be the array of the projection of $\loadmatrix$ among the
  second dimension inside a given interval $c$ of processors: $S[j] =
  \sum_{i \in c} \loadmatrix [i][j]$. The optimal partition of $S$
  respects: $\maxload^*(S) \leq \frac{\sum_j S[j]}{Q} (1+ \Delta
  \frac{Q}{\sizeY})$. Since $S$ is given by the partition of $C$, we
  have $\sum_j S[j] \leq \maxload^*(C)$ which leads to $\maxload^*(S)
  \leq (1+\Delta \frac{P}{\sizeX})(1+\Delta \frac{Q}{\sizeY})
  \frac{\sum_{i,j} \loadmatrix [i][j]}{PQ}$
\end{proof}

It remains the question of the choice of $P$ and $Q$ which is solved
by the following theorem.

\begin{theorem}
  The approximation ratio of {\tt JAG-PQ-HEUR} is minimized by $P =
  \sqrt{\nbproc\frac{\sizeX}{\sizeY}}$.
\end{theorem}

\begin{proof}
  The approximation ratio of {\tt JAG-PQ-HEUR} can be written as $f(x)
  = (1+ax)(1+b/x)$ with $a,b,x>0$ by setting $a =
  \frac{\Delta}{\sizeX}$, $b = \frac{\Delta \nbproc}{\sizeY}$ and $x =
  P$. The minimum of $f$ is now computed by studying its derivative:
  $f'(x) = a- b/x^2$.  $f'(x) < 0 \iff x < \sqrt{b/a}$ and $f'(x) > 0
  \iff x > \sqrt{b/a}$. It implies that $f$ has one minimum given by
  $f'(x) = 0 \iff x=\sqrt{b/a}$.
\end{proof}

Notice that when $\sizeX=\sizeY$, the approximation ratio is minimized
by $P=Q=\sqrt\nbproc$.\\

Two algorithms exist to find an optimal \pxq-way jagged partition in
polynomial time. The first one, we call {\tt JAG-PQ-OPT-NICOL}, has
been proposed first by~\cite{Pinar97} and is constructed by using the
1D algorithm presented in~\cite{Nicol94}. This algorithm is of
complexity $O((P Q \log \frac{\sizeX}{P} \log \frac{\sizeY}{Q})^2
)$. The second one, we call {\tt JAG-PQ-OPT-DP} is a dynamic
programming algorithm proposed by~\cite{Manne96}. Both algorithms
partition the main dimension using a 1D partitioning algorithm using
an optimal partition of the auxiliary dimension for the evaluation of
the load of an interval. The complexity of {\tt JAG-PQ-OPT-DP} is
$O(\sizeX \log \sizeX (P + ( Q \log \frac{\sizeY}{Q})^2))$.

\subsubsection{\nbproc-way Jagged Partitions}
\label{sec:m-way-jagged}

We introduce the notion of \nbproc-way jagged partitions which allows
jagged partitions with different numbers of processors in each
interval of the main dimension. Indeed, even the optimal partition in
the main dimension may have a high load imbalance and allocating more
processor to one interval might lead to a better load balance. Such a
partition is presented in Figure~\ref{fig:jagged_part_mway}. We
propose four algorithms to generate \nbproc-way jagged partitions. The
first one is {\tt JAG-M-HEUR}, a heuristic extending the \pxq-way
jagged partitioning heuristic. The second algorithm generates the
optimal \nbproc-way jagged partition for given intervals in the main
dimension, leading to {\tt JAG-M-HEUR-PROBE}. Then, the third
algorithm, called {\tt JAG-M-ALLOC}, generates the optimal \nbproc-way jagged
partition for a given number of interval provided the number of
processor inside each interval is known. Finally, we present {\tt
  JAG-M-OPT}, a polynomial optimal dynamic programming algorithm.\\

We propose {\tt JAG-M-HEUR} which is a heuristic similar to {\tt
  JAG-PQ-HEUR}. The main dimension is first partitioned in $P$
intervals using an optimal 1D partitioning algorithm which define $P$
stripes. Then each stripe $S$ is allocated a number of processors
$Q_S$ which is proportional to the load of the interval. Finally, each
interval is partitioned on the auxiliary dimension using $Q_S$
processors by an optimal 1D partitioning algorithm.

Choosing $Q_S$ is a non trivial matter since distributing the
processors proportionally to the load may lead to non integral values
which might be difficult to round. Therefore, we only distribute
proportionally $(\nbproc - P)$ processors which allows to round the
allocation up: $Q_S = \Ceil{(\nbproc - P) \frac{\sum_{i,j \in S}
    \loadmatrix[i][j]}{\sum_{i,j} \loadmatrix [i][j]}}$. Notice that
between $0$ and $P$ processors remain unallocated. They are allocated,
one after the other, to the interval that maximizes $\frac{\sum_{i,j \in
    S}\loadmatrix[i][j]}{Q_S}$. 

An analysis of the performance of {\tt JAG-M-HEUR} similar to the one
proposed for {\tt JAG-PQ-HEUR} that takes the distribution of the
processors into account is now provided.

\begin{theorem}
  \label{th:m-way-guarantee}
  If there is no zero in \loadmatrix, {\tt JAG-M-HEUR} is a
  $(\frac{\nbproc}{\nbproc-P} +
  \frac{\nbproc \Delta}{P \sizeY} + \frac{\Delta^2
    \nbproc}{\sizeX\sizeY})$-approximation algorithm where $\Delta =
  \frac{\max_{i,j} \loadmatrix [i][j]}{\min_{i,j} \loadmatrix
    [i][j]}$, $P < \sizeX$.
\end{theorem}

\begin{proof}
  Let $C$ denote the array of the projection of $A$ among one
  dimension: $C[i] = \sum_j \loadmatrix[i][j]$. Similarly to the
  proof of Theorem~\ref{th:pq-heur-approx}, we have: $\maxload^*(C)
  \leq \frac{\sum \loadmatrix [i][j]}{P} (1+ \Delta \frac{P}{\sizeX})$

  Let $S$ denote the array of the projection of $\loadmatrix$ among
  the second dimension inside a given interval $c$ of an optimal partition of
  $C$. $S[j] = \sum_{i \in c} \loadmatrix [i][j]$. We have $\sum_{j}
  S[j] \leq \maxload^*(C)$. Then, the number of processors allocated
  to the stripe is bounded by: $\frac{(\nbproc - P )
    \sum_j S[j] }{\sum_{i,j} \loadmatrix [i][j]} \leq Q_S \leq
  \frac{(\nbproc - P ) \sum_j S[j] }{\sum_{i,j} \loadmatrix [i][j]} +
  1$. The bound on $\sum_{j} S[j]$ leads to $Q_S \leq \frac{
    \nbproc-P }{P } (1+ \frac{\Delta P}{\sizeX})+1$.

  We now can compute bounds on the optimal partition of stripe
  $S$. The bound from Lemma~\ref{lem:DC-approx} states: $\maxload^*(S)
  \leq \frac{\sum_j S[j] }{Q_S}  (1+ \frac{\Delta
    Q_S}{\sizeY})$. The bounds on $\sum_j S[j]$ and $Q_S$
  imply $\maxload^*(S) \leq \frac{\sum \loadmatrix [i][j]}{\nbproc}
  (\frac{\nbproc}{\nbproc-P} + \frac{\nbproc}{P}\frac{\Delta}{\sizeY}
  + \frac{\Delta^2 \nbproc}{\sizeX\sizeY})$.

  The load imbalance (and therefore the approximation ratio) is less
  than $(\frac{\nbproc}{\nbproc-P} + \frac{\nbproc}{P}
  \frac{\Delta}{\sizeY} + \frac{\Delta^2 \nbproc}{\sizeX\sizeY})$.
\end{proof}

This approximation ratio should be compared to the one obtained by
{\tt JAG-PQ-HEUR} which can be rewritten as $((1+\Delta
\frac{P}{\sizeX})+ \frac{ \Delta \nbproc}{P \sizeY}+\frac{\Delta^2
  \nbproc}{\sizeX \sizeY})$. Basically, using \nbproc-way partitions
trades a factor of $(1+\frac{P \Delta}{\sizeX})$ to the profit of a
factor $\frac{\nbproc}{\nbproc-P}$. 

We can also compute the number of stripes $P$ which optimizes the
approximation ratio of {\tt JAG-M-HEUR}.

\begin{theorem}\label{th:m-way-opt-p}
  The approximation ratio of {\tt JAG-M-HEUR} is minimized by
  $P = \frac{\sqrt{\Delta^2 (\nbproc^2 - 1) - \sizeY} - \Delta
    \nbproc}{\sizeY - \Delta}$.
\end{theorem}

\begin{proof}
  We analyze the function of the approximation ratio in function of
  the number of stripes: $f(P) = (\frac{\nbproc}{\nbproc-P} +
  \frac{\nbproc}{P} \frac{\Delta}{\sizeY} + \frac{\Delta^2
    \nbproc}{\sizeX\sizeY})$. Its derivative is: $f'(P) =
  \frac{\nbproc}{(\nbproc-P)^2}-\frac{\nbproc \Delta}{\sizeY
    P^2}$. The derivative is negative when $P$ tends to $0^+$,
  positive when $P$ tends to $+\infty$ and null when $(\sizeY -
  \Delta) P^2 + 2 \nbproc \Delta P - \Delta \nbproc^2 = 0$. This
  equation has a unique positive solution: $P = \frac{\sqrt{\Delta^2
      (\nbproc^2 - 1) - \sizeY} - \Delta \nbproc}{\sizeY - \Delta}$.
\end{proof}

This result is fairly interesting. The optimal number of stripes
depends of $\Delta$ and depends of \sizeY but not of \sizeX. The
dependency of $\Delta$ makes the determination of $P$ difficult in
practice since a few extremal values may have a large impact on the
computed $P$ without improving the load balance in
practice. Therefore, {\tt JAG-M-HEUR} will use $\sqrt{\nbproc}$
stripes. The complexity of {\tt JAG-M-HEUR} is $O((P \log
\frac{\sizeX}{P})^2 + \sum_S (Q_S \log \frac{\sizeY}{Q_S})^2 )$ which
in the worst case is $O((P \log
\frac{\sizeX}{P})^2 + (\nbproc \log \frac{\sizeY}{\nbproc})^2 )$
\\

We now explain how one can build the optimal jagged partition provided
the partition in the main dimension is given. This problem reduces to
partitioning $P$ one dimensional arrays using \nbproc processors in
total to minimize \maxload. \cite{Han92}~states that the proposed
algorithms apply in the presence of multiple chains but does not
provide much detail. We explain how to extend {\tt
  NicolPlus}~\cite{Pinar04} to the case of multiple one dimensional
arrays.

We now first explain the algorithm {\tt PROBE-M} for partitioning
multiple arrays that test the feasibility of a given maximum load
\maxload.

The main idea behind {\tt PROBE-M} is to compute for each one
dimensional array how many processors are required to achieve a
maximum load of \maxload. For one array, the number of processors
require to achieve a load of \maxload is obtained by greedily
allocating intervals maximal by inclusion of load less than
\maxload. The boundary of these intervals can be found in $O(\log
\sizeOneD)$ by a binary search. Across all the arrays, there is no
need to compute the boundaries of more than \nbproc intervals, leading
to an algorithm of complexity $O(\nbproc \log \sizeOneD)$.

\cite{Han92} reduces the complexity of the one dimensional partitioning
problem to $O(\nbproc \log \frac{\sizeOneD}{\nbproc})$ by slicing the
array in \nbproc chunks. That way, one has first to determine in which
chunk the borders of the intervals are, and then perform a binary
search in the chunk. Provided there are \nbproc intervals to generate,
the cost of selecting the right chunk is amortized. But it does not
directly apply to the multiple array partitioning problem. Indeed,
slicing the array in such a manner will lead to a complexity of
$O(\nbproc \log \frac{\sizeOneD}{\nbproc} + P \nbproc)$. However,
slicing the arrays in chunk of size $\frac{\sizeOneD P}{\nbproc}$
leads to having at most $\nbproc+P$ chunks. Therefore, {\tt PROBE-M}
has a complexity of $O(\nbproc \log \frac{\sizeOneD P}{\nbproc} +
\nbproc + P) = O(\nbproc \log \frac{\sizeOneD P}{\nbproc})$.

Notice that the engineering presented in~\cite{Pinar04} for the single
array case that use an upper bound and a lower bound on the position
of each boundary can still be used when there are multiple arrays with
{\tt PROBE-M}.  When the values of \maxload decreases, the $i$th cut
inside one array is only going to move toward the beginning of the
array. Conversely, when \maxload increase, the $i$th cut inside one
array is only going to move toward the end of the array. One should
notice that the number of processors allocated to one array might vary
when \maxload varies.

With {\tt PROBE-M}, one can solve the multiple array partitioning
problem in multiple way. An immediate one is to perform a binary
search on the values of \maxload. It is also possible to reuse the
idea of {\tt NicolPlus} which rely on the principle that the first
interval is either maximal such that the load is an infeasible maximum
load or minimal such that the load is a feasible maximum load. The
same idea applies by taking the intervals of each array in the same
order {\tt PROBE-M} considers them. The windowing trick still applies
and leads to an algorithm of complexity $O((\nbproc \log
\frac{\sizeOneD P}{\nbproc})^2)$. Given the stripes in the main
dimension, {\tt JAG-M-PROBE} is the algorithm that applies the
modified version of {\tt NicolPlus} to generate an \nbproc-way
partition.

Other previous algorithms apply to this problem. For instance,
\cite{Bokhari88}~solves the multiple chains problem on host-satellite
systems. One could certainly use this algorithm but the runtime
complexity is $O(\sizeOneD^3 \nbproc \log \sizeOneD)$. Another way to
solve the problem can certainly be derived from the work of
Frederickson~\cite{Frederickson1990}.

{\tt JAG-M-HEUR-PROBE} is the algorithm that uses the stripes obtained
by {\tt JAG-M-HEUR} and then applies {\tt JAG-M-PROBE}.
\\

Given a number of stripes $P$ and the number of processors $Q_S$
inside each stripe, one can compute the optimal \nbproc-way jagged
partition. The technique is similar to the optimal \pxq-way jagged
partitioning technique shown in~\cite{Pinar97}. {\tt NicolPlus} gives
an optimal partition not only on one dimensional array but on any one
dimension structure where the load of intervals are monotonically
increasing by inclusion. When {\tt NicolPlus} needs the load of an
interval, one can return the load of the optimal $Q_S$-way partition
of the auxiliary dimension, computed in $O((Q_S \log
\frac{\sizeY}{Q_S})^2)$.

To generate the \nbproc-way partition, one needs to modify {\tt
  NicolPlus} to keep track of which stripe an interval represents to
return the load of the optimal partition of the auxiliary dimension
with the proper number of processors. This modification is similar to
using {\tt NicolPlus} to solve the heterogeneous array partitioning
problem~\cite{Pinar08}. Let us call this algorithm {\tt JAG-M-ALLOC}.
The overall algorithm has a complexity of $O((P \log \frac{\sizeX}{P}
\max_S Q_S \log \frac{\sizeY}{Q_S})^2)$.
\\

We provide another algorithm, {\tt JAG-M-OPT} which builds an
optimal \nbproc-way jagged partition in polynomial time using
dynamic programming. An optimal solution can be represented by $k$,
the beginning of the last interval on the main dimension, and $x$, the
number of processors allocated to that interval. What remains is a
$(\nbproc-x)$-way partitioning problem of a matrix of size $(k-1)
\times \sizeY$. It is obvious that the interval $\{(k-1), \dots,
\sizeX\}$ can be partitioned independently from the remaining
array. The dynamic programming formulation is:
$$\maxload (\sizeX, \nbproc)  = \min_{1 \leq k \leq \sizeX, 1 \leq x
  \leq \nbproc} \max \{ \maxload(k-1,\nbproc-x), 1D(k,\sizeX,x)\}$$

where $1D(i,j,k)$ denotes the value of the optimal 1D partition
among the auxiliary dimension of the $[i,j]$ interval on $k$ processors.

There are at most $\sizeX \nbproc$ calls to \maxload to evaluate, and
at most $\sizeX^2 \nbproc$ calls to $1D$ to evaluate. Evaluating one
function call of \maxload can be done in $O(\sizeX \nbproc)$ and
evaluating $1D$ can be done in $O((x \log \frac{\sizeY}{x})^2)$ using the
algorithm from~\cite{Nicol94}. The algorithm can trivially be
implemented in $O((\sizeX \nbproc)^2 + \sizeX^2 \nbproc^3 (\log \frac{
  \sizeY }{ \nbproc })^2) = O(\sizeX^2 \nbproc^3 (\log \frac{ \sizeY
}{ \nbproc })^2)$ which is polynomial.

However, this complexity is an upper bound and several improvements can
be made, allowing to gain up to two orders of
magnitude in practice. First of all, the different values of both functions
\maxload and $1D$ can only be computed if needed. Then the parameters
$k$ and $x$ can be found using binary search. For a given $x$,
$\maxload(k-1,\nbproc-x)$ is an increasing function of $k$, and
$1D(k,\sizeX,x)$ is a decreasing function of $k$. Therefore, their
maximum is a bi-monotonic, decreasing first, then increasing function
of $k$, and hence its minimum can be found using a binary search.

Moreover, the function $1D$ is the value of an optimal 1D partition, and we know
lower bounds and an upper bound for this function. Therefore, if
$\maxload(k-1, \nbproc-x) > UB(1D(k,\sizeX,x))$, there is no need to
evaluate function $1D$ accurately since it does not give the
maximum. Similar arguments on lower and upper bound of $\maxload(k-1,
\nbproc-x)$ can be used.

Finally, we are interested in building an optimal \nbproc-way jagged
partition and we use branch-and-bound techniques to speed up the
computation. If we already know a solution to that problem (Initially
given by a heuristic such as {\tt JAG-M-HEUR} or found during the
exploration of the search space), we can use its maximum load $l$ to
decide not to explore some of those functions, if the values (or their
lower bounds) $\maxload$ or $1D$ are larger than $l$.

\subsection{Hierarchical Bipartition}
\label{sec:hierar_part}

Hierarchical bipartitioning techniques consist of obtaining partitions
that can be recursively generated by splitting one of the dimensions in
two intervals. An example of such a partition is depicted in
Figure~\ref{fig:hierar_part}. Notice that such partitions can be
represented by a binary tree for easy indexing. We present first
{\tt HIER-RB}, a known algorithm to generate hierarchical
bipartitions. Then we propose {\tt HIER-OPT}, an original optimal
dynamic programming algorithm. Finally, a heuristic algorithm, called
{\tt HIER-RELAXED} is derived from the dynamic programming
algorithm.\\

A classical algorithm to generate hierarchical bipartition is Recursive
Bisection which has originally been proposed in~\cite{Berger87} and
that we call in the following {\tt HIER-RB}. It cuts the matrix into
two parts of (approximately) equal load and allocates half the
processors to each sub-matrix which are partitioned recursively. The
dimension being cut in two intervals alternates at each level of the
algorithm.  This algorithm can be implemented in $O(\nbproc \log \max
(\sizeX, \sizeY))$ since finding the position of the cut can be done
using a binary search.

The algorithm was originally designed for a number of processors which
is a power of 2 so that the number of processors at each step is
even. However, if at a step the number of processors is odd, one part
will be allocated $\Floor{\frac{\nbproc}{2}}$ processors and the other
part $\Floor{\frac{\nbproc}{2}} + 1$ processors. In such a case, the
cutting point is selected so that the load per processor is minimized.

Variants of the algorithm exist based on the decision of the dimension
to partition. One variant does not alternate the partitioned dimension
at each step but virtually tries both dimensions and selects the one
that lead to the best expected load balance~\cite{Vastenhouw05}. Another
variant decides which direction to cut by selecting the direction with
longer length.\\

We now propose {\tt HIER-OPT}, a polynomial algorithm for generating
the optimal hierarchical partition. It uses dynamic programming and
relies on the tree representation of a solution of the problem. An
optimal hierarchical partition can be represented by the orientation
of the cut, the position of the cut (denoted $x$ or $y$, depending on
the orientation), and the number of processors $j$ in the first part.

The algorithm consists in evaluating the function $\maxload (x_1, x_2,
y_1, y_2 ,\nbproc)$ that partitions rectangle $(x_1, x_2, y_1, y_2 )$ using
$\nbproc$ processors. 
\begin{flalign}
  \maxload(x_1, x_2,  y_1, y_2, &\nbproc)=\min_j \min \big\{ \\
   &\min_x \max \{\maxload(x_1, x, y_1, y_2 ,j), \label{eq:firstdimp1}\\
   &\maxload(x+1, x_2, y_1, y_2 , m-j)\}, \label{eq:firstdimp2}\\
   &\min_y \max \{\maxload(x_1, x_2, y_1, y, j),\label{eq:seqdimp1}\\
   &\maxload(x_1, x_2, y+1, y_2, m-j)\} \big\} \label{eq:seqdimp2} 
\end{flalign}

Equations~\ref{eq:firstdimp1} and~\ref{eq:firstdimp2}
consider the partition in the first dimension and
Equations~\ref{eq:seqdimp1} and ~\ref{eq:seqdimp2} consider it in the
second dimension. The dynamic programming provides the position $x$ (or
$y$) to cut and the number of processors ($j$ and $\nbproc-j$) to
allocate to each part.

This algorithm is polynomial since there are $O(\sizeX^2 \sizeY^2
\nbproc)$ functions $\maxload$ to evaluate and each function can
\naively be evaluated in $O((x_2-x_1+y_2-y1) \nbproc)$. Notice that
optimization techniques similar to the one used in
Section~\ref{sec:m-way-jagged} can be applied. In particular $x$ and
$y$ can be computed using a binary search reducing the complexity of
the algorithm to $O(\sizeX^2 \sizeY^2 \nbproc^2 \log(\max
(\sizeX,\sizeY))) )$.\\

Despite the dynamic programming formulation is polynomial, its
complexity is too high to be useful in practice for real sized
systems. We extract a heuristic called {\tt HIER-RELAXED}. To
partition a rectangle $(x_1,x_2,y_1,y_2)$ on \nbproc processors, {\tt
  HIER-RELAXED} computes the $x$ (or $y$) and $j$ that optimize the
dynamic programming equation, but substitutes the recursive calls to
$\maxload()$ by a heuristic based on the average load: That is to say,
instead of making recursive $\maxload(x, x', y, y' ,j)$ calls, $
\frac{\load(x, x', y, y')}{j}$ will be calculated. The values of $x$
(or $y$) and $j$ provide the position of the cut and the number of
processors to allocate to each part respectively. Each part is
recursively partitioned. The complexity of this algorithm is
$O(\nbproc^2 \log(\max (\sizeX,\sizeY))) )$.

\subsection{More General Partitioning Schemes}

The considerations on Hierarchical Bipartition can be extended to any
kind of recursively defined pattern such as the ones presented in
Figures~\ref{fig:rec_part1} and~\ref{fig:rec_part2}. As long as there
are a polynomial number of possibilities at each level of the
recursion, the optimal partition following this rule can be generated
in polynomial time using a dynamic programming technique. The number
of functions to evaluate will keep being in $O(\sizeX^2 \sizeY^2
\nbproc)$; one function for each sub rectangle and number of
processors.. The only difference will be in the cost of evaluating the
function calls. In most cases if the pattern is composed of $k$
sections, the evaluation will take $O((\max (\sizeX, \sizeY)
\nbproc)^{k-1})$.

This complexity is too high to be of practical use but it proves that
an optimal partition in these classes can be generated in polynomial
time. Moreover, those dynamic programming can serve as a basis to
derive heuristics similarly to {\tt HIER-RELAXED}.

A natural question is ``given a maximum load, is it possible to
compute an arbitrary rectangular partition?'' \cite{Khanna98}~shows
that such a problem is NP-Complete and that there is no approximation
algorithm of ratio better than $\frac{5}{4}$ unless P=NP. Recent
work~\cite{Paluch06} provides a 2-approximation algorithm which
heavily relies on linear programming.

\section{Experimental Evaluation}
\label{sec:expe}

\subsection{Experimental Setting}

This section presents an experimental study of the main
algorithms. For rectilinear partitions, both the uniform partitioning
algorithm {\tt RECT-UNIFORM} and {\tt RECT-NICOL} algorithm have been
implemented. For \pxq-way and $\nbproc$-way jagged partitions, the
following heuristics and optimal algorithms have been implemented:
{\tt JAG-PQ-HEUR}, {\tt JAG-PQ-OPT-NICOL}, {\tt JAG-PQ-OPT-DP}, {\tt
  JAG-M-HEUR}, {\tt JAG-M-HEUR-PROBE} and {\tt JAG-M-OPT}. Each jagged
partitioning algorithm exists in three variants, namely {\tt -HOR}
which considers the first dimension as the main dimension, {\tt -VER}
which considers the second dimension as the main dimension, and {\tt
  -BEST} which tries both and selects the one that leads to the best
load balance. For hierarchical partitions, both recursive bisection
{\tt HIER-RB} and the heuristic {\tt HIER-RELAXED} derived from the
dynamic programming have been implemented. Each hierarchical
bipartition algorithm exists in four variants {\tt -LOAD} which
selects the dimension to partition according to get the best load,
{\tt -DIST} which partitions the longest dimension, and {\tt -HOR} and
{\tt -VER} which alternate the dimension to partition at each level of
the recursion and starting with the first or the second dimension.

The algorithms were tested on the BMI department cluster called
Bucki. Each node of the cluster has two 2.4 GHz AMD Opteron(tm)
quad-core processors and 32GB of main memory. The nodes run on Linux
2.6.18. The sequential algorithms are implemented in C++. The
compiler is g++ 4.1.2 and -O2 optimization was used.

The algorithms are tested on different classes of instances. Some are
synthetic and some are extracted from real applications. The first set of
instances is called PIC-MAG. These instances are extracted from the
execution of a particle-in-cell code which simulates the interaction
of the solar wind on the Earth's magnetosphere~\cite{Karimabadi06}. In
those applications, the computational load of the system is mainly
carried by particles.  We extracted the distribution of the
particles every 500 iterations of the simulations for the first 33,500
iterations. These data are extracted from a 3D simulation. Since the
algorithms are written for the 2D case, in the PIC-MAG instances, the
number of particles are accumulated among one dimension to get a 2D
instance. A PIC-MAG instance at  iteration 20,000 can be seen in
Figure~\ref{fig:homa3d_inst}. The intensity of a pixel is linearly
related to computation load for that pixel (the whiter the more
computation). During the course of the simulation, the
particles move inside the space leading to values of $\Delta$ varying
between $1.21$ and $1.51$.

\begin{figure}
  \centering
  \subfigure[PIC-MAG]{\includegraphics[width=0.3\linewidth]{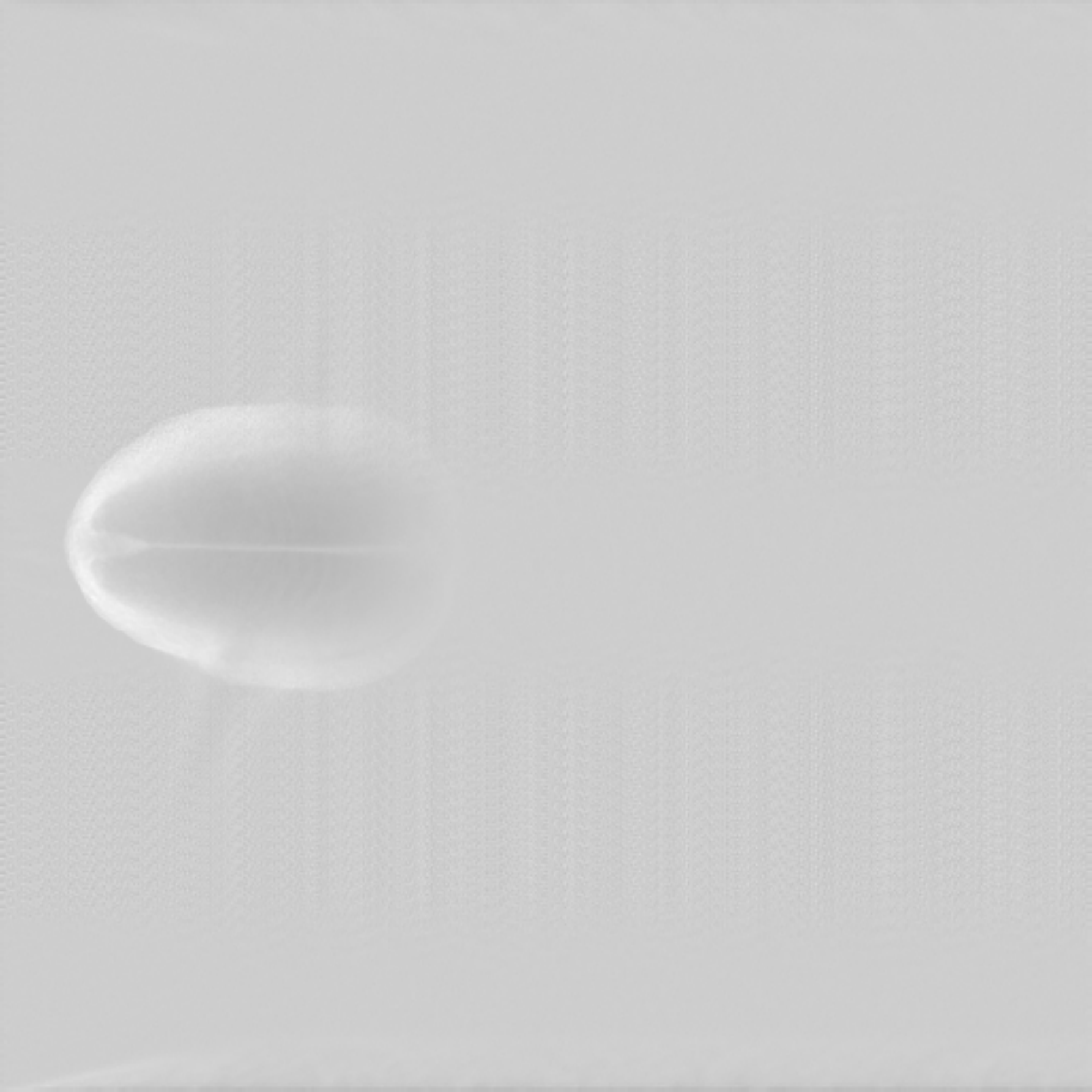} \label{fig:homa3d_inst}}
  \subfigure[SLAC]{\includegraphics[width=0.3\linewidth]{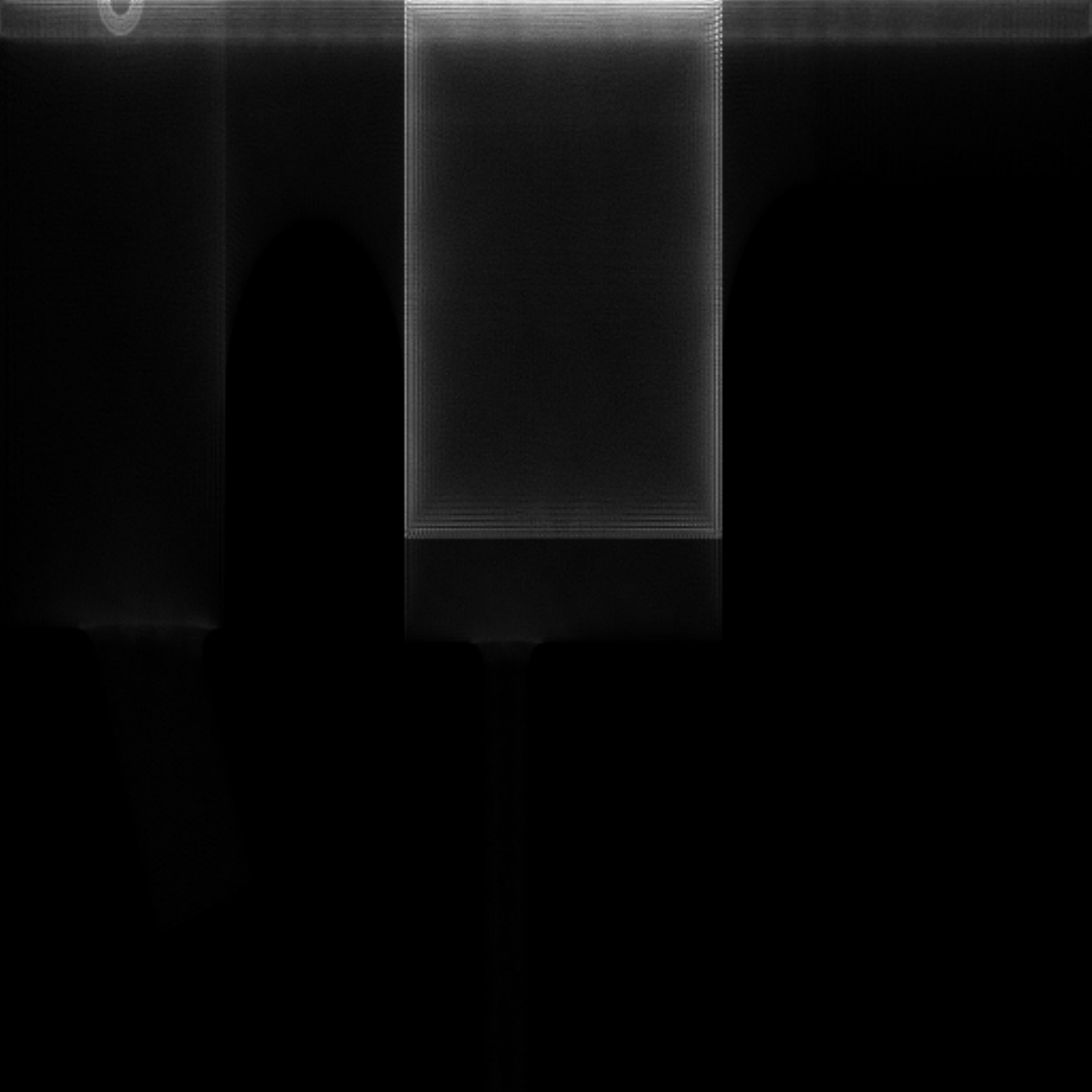}  \label{fig:SLAC_inst}}
  \subfigure[Diagonal]{\includegraphics[width=0.3\linewidth]{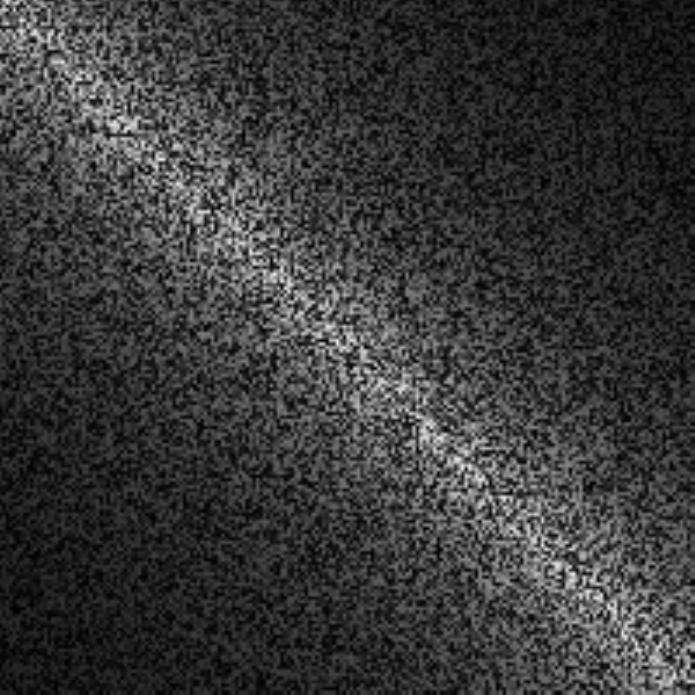} \label{fig:diag_inst}} 
  \subfigure[Peak]{\includegraphics[width=0.3\linewidth]{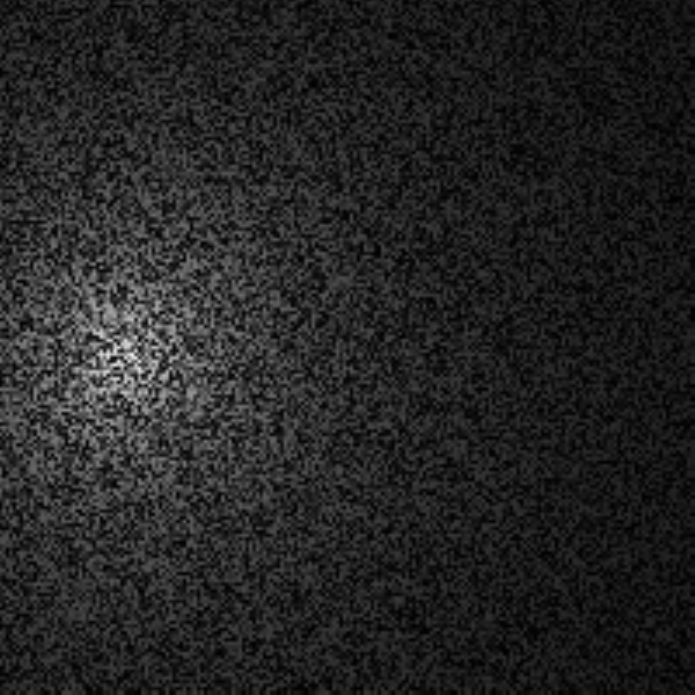} \label{fig:peak_inst}}
  \subfigure[Multi-peak]{\includegraphics[width=0.3\linewidth]{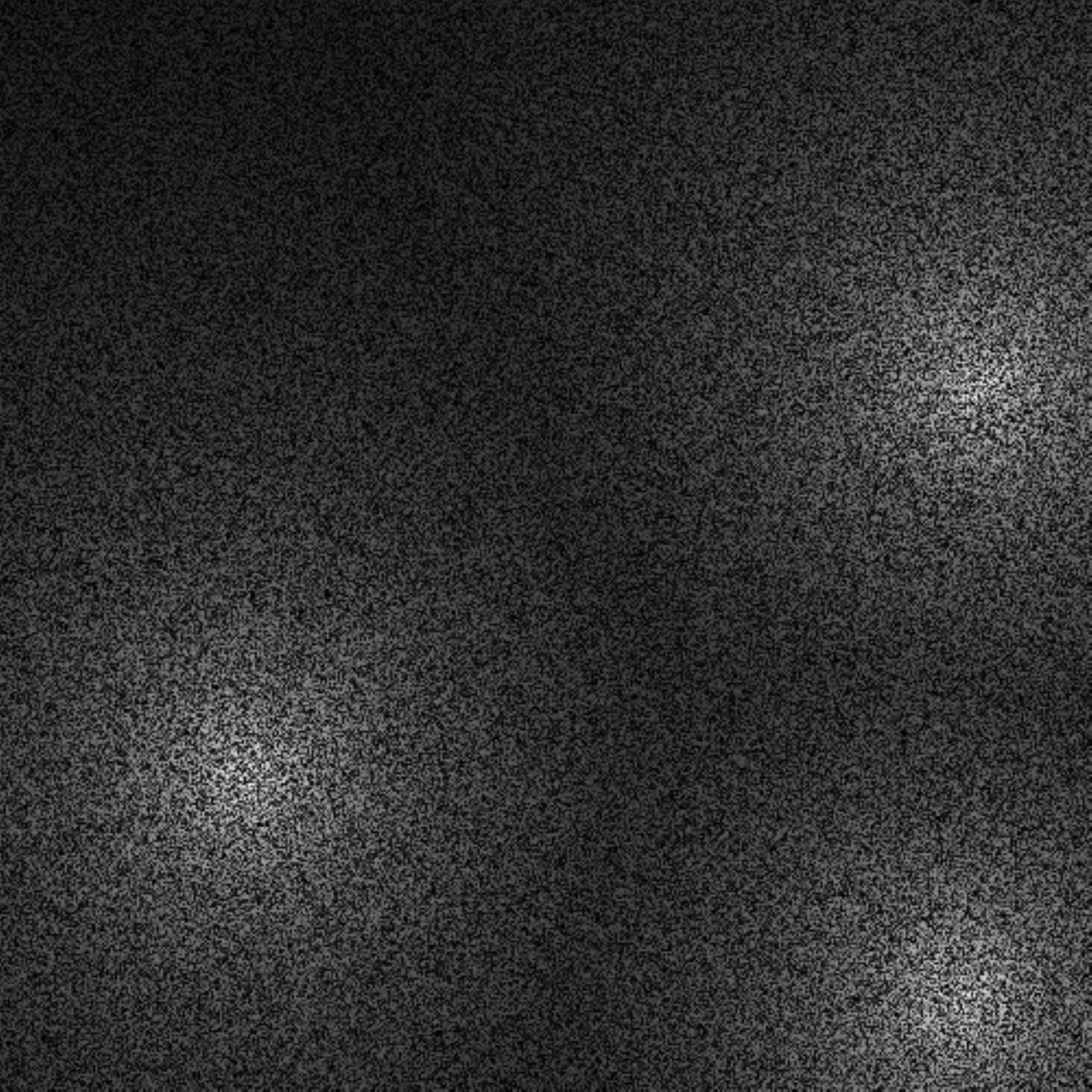} \label{fig:mpeak_inst}}
  \subfigure[Uniform]{\includegraphics[width=0.3\linewidth]{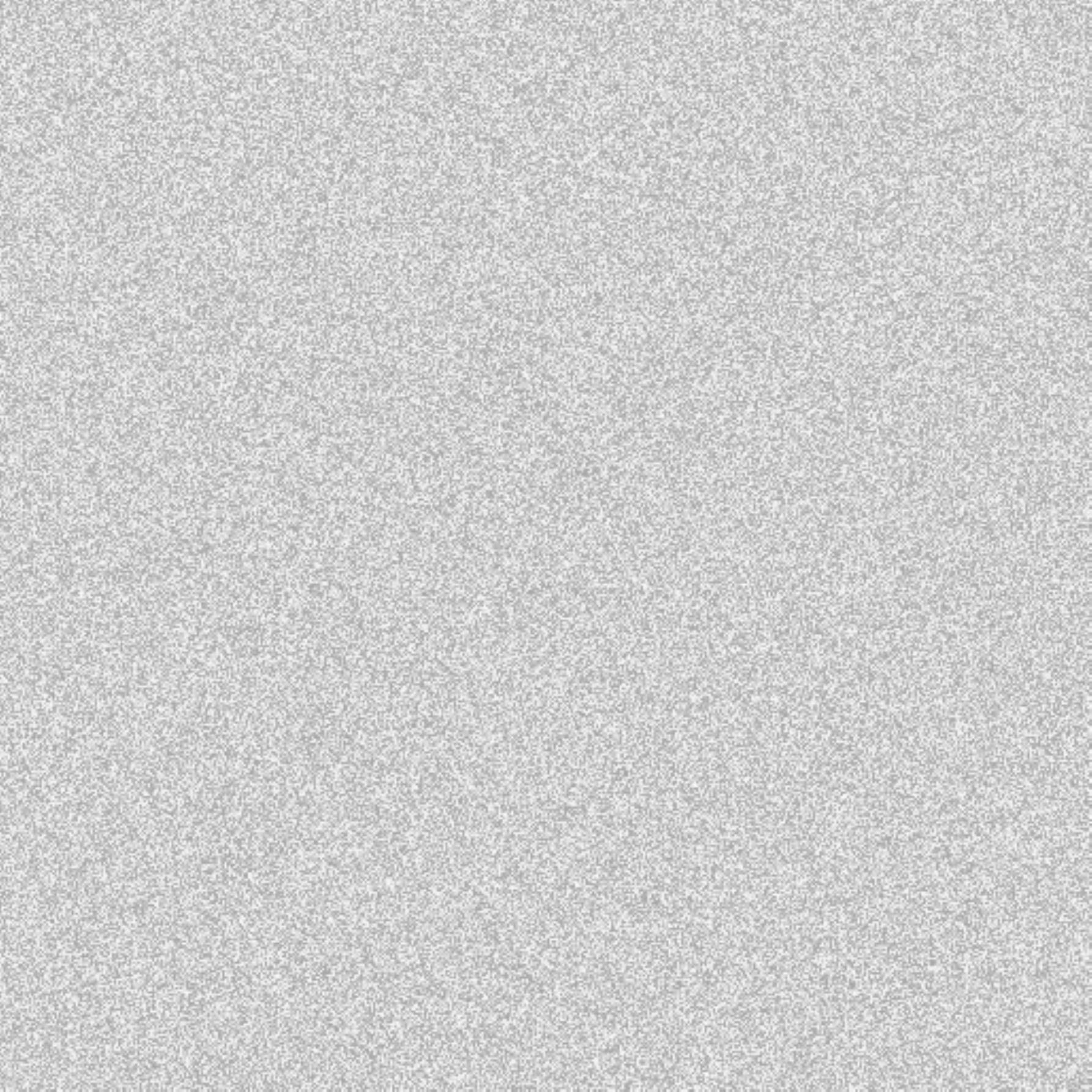} \label{fig:unif_inst}}
  \caption{Examples of real and synthetic instances.}
\end{figure}

The SLAC dataset (depicted in Figure~\ref{fig:SLAC_inst}) is generated from the mesh of a 3D object. Each
vertex of the 3D object carries one unit of computation. Different
instances can be generated by projecting the mesh on a 2D plane and by
changing the granularity of the discretization. This setting match
the experimental setting of~\cite{Nicol94}. In the experiments, we
generated instances of size 512x512. Notice that the matrix contains
zeroes, therefore $\Delta$ is undefined.

Different classes of synthetic squared matrices are also used, these
classes are called diagonal, peak, multi-peak and uniform. Uniform
matrices (Figure~\ref{fig:unif_inst}) are generated to obtain a given
value of $\Delta$: the computation load of each cell is generated
uniformly between $1000$ and $1000*\Delta$. In the other three
classes, the computation load of a cell is given by generating a
number uniformly between 0 and the number of cells in the matrix which
is divided by the Euclidean distance to a reference point (a 0.1
constant is added to avoid dividing by zero). The choice of the
reference point is what makes the difference between the three classes
of instances. In diagonal (Figure~\ref{fig:diag_inst}), the reference
point is the closest point on the diagonal of the matrix. In peak
(Figure~\ref{fig:peak_inst}), the reference point is one point chosen
randomly at the beginning of the execution. In multi-peak
(Figure~\ref{fig:mpeak_inst}), several points (here 3) are randomly
generated and the closest one will be the reference point. Those
classes are inspired from the synthetic data from~\cite{Manne96}.

The performance of the algorithms is given using the load imbalance
metric defined in Section~\ref{sec:model}. For synthetic dataset,
the load imbalance is computed over 10 instances as follows:
$\frac{\sum_{I}{\maxload}(I)}{\sum_{I}{\load_{avg}}(I)} - 1$. The
experiments are run on most square number of processors between 16 and
10,000. Using only square numbers allows us to fix the parameter
$P=\sqrt{\nbproc}$ for all rectilinear and jagged algorithms.

\subsection{Jagged algorithms}

The jagged algorithms have three variants, two depending on whether the
main dimension is the first one or the second one and the third tries
both of them and takes the best solution. On all the fairly homogeneous
instances (i.e., all but the mesh SLAC), the load imbalance of the three
variants are quite close and the orientation of the jagged partitions
does not seem to really matter. However this is not the same in
\nbproc-way jagged algorithms where the selection of the main dimension
can make significant differences on overall load imbalance. Since
the \nbproc-way jagged partitioning heuristics are as fast as heuristic jagged
partitioning, trying both dimensions and taking the one with best load
imbalance is a good option. From now on, if the variant of a jagged
partitioning algorithm is unspecified, we will refer to their {\tt BEST}
variant.

We proposed in Section~\ref{sec:m-way-jagged} a new type of jagged
partitioning scheme, namely, \nbproc-way jagged, which does not
require all the slices of the main dimension to have the same number
of processors. This constraint is artificial in most cases and we show
that it significantly harms the load balance of an application.

\begin{figure}
  \centering
  \includegraphics[width=.8\linewidth]{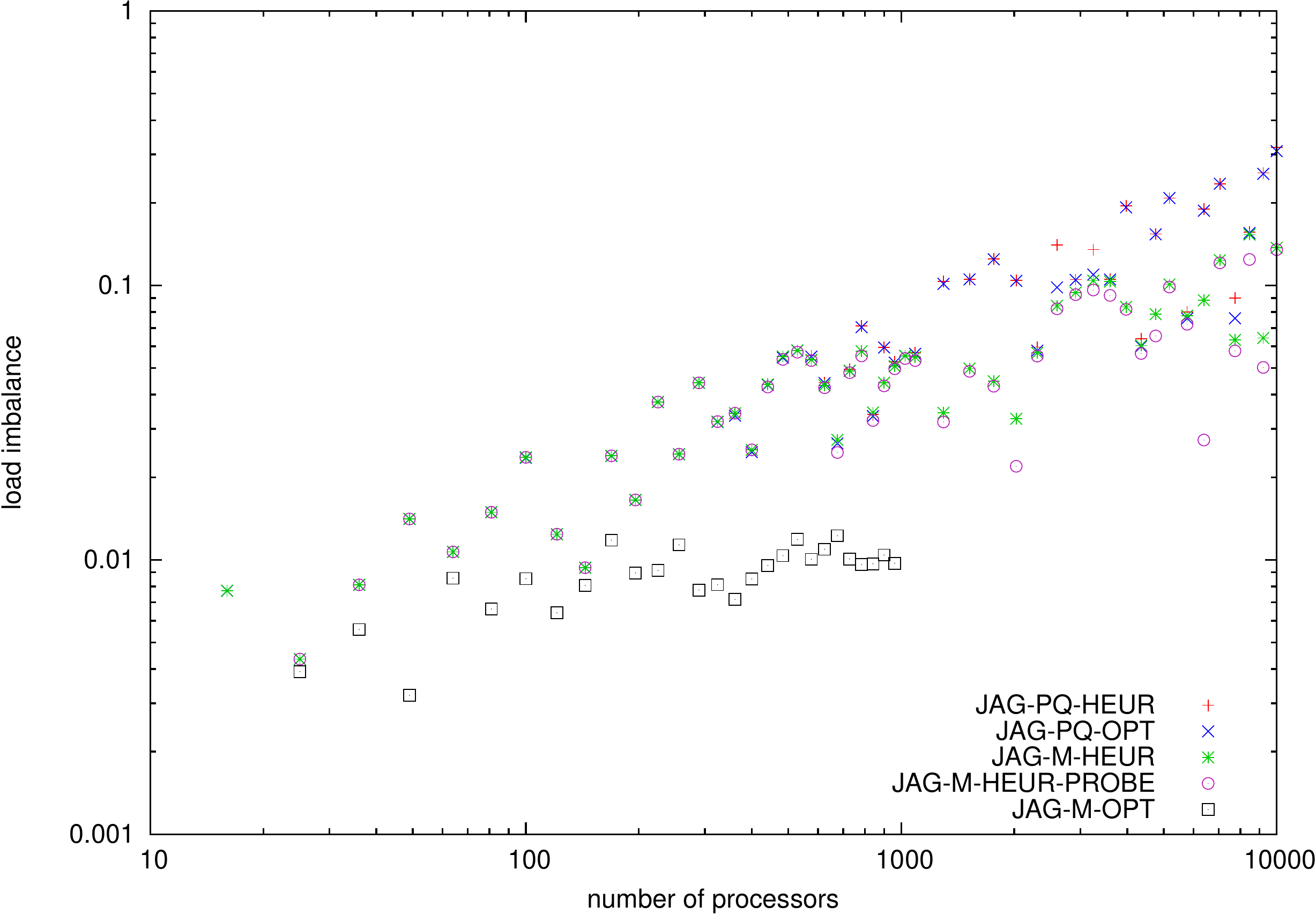}
  \caption{Jagged methods on PIC-MAG iter=30,000.}
  \label{fig:homa3d_it30000_jagged}
\end{figure}

Figure~\ref{fig:homa3d_it30000_jagged} presents the load balance
obtained on PIC-MAG at iteration 30,000 with heuristic and optimal
\pxq-way jagged algorithms and \nbproc-way jagged algorithms. On less
than one thousand processors, {\tt JAG-M-HEUR}, {\tt JAG-PQ-HEUR},
{\tt JAG-PQ-OPT} and {\tt JAG-M-HEUR-PROBE} produce almost the same
results (hence the points on the chart are super imposed). Note that,
{\tt JAG-PQ-HEUR} and {\tt JAG-PQ-OPT} obtain the same load imbalance
most of the time even on more than one thousand processors. This
indicates that there is almost no room for improvement for the
\pxq-way jagged heuristic. {\tt JAG-M-HEUR-PROBE} usually obtains the
same load imbalance that {\tt JAG-M-HEUR} or does slightly better. But
on some cases, it leads to dramatic improvement. One can remark that
the \nbproc-way jagged heuristics always reaches a better load balance
than the \pxq-way jagged partitions.

\begin{figure}
  \centering
  \includegraphics[width=.8\linewidth]{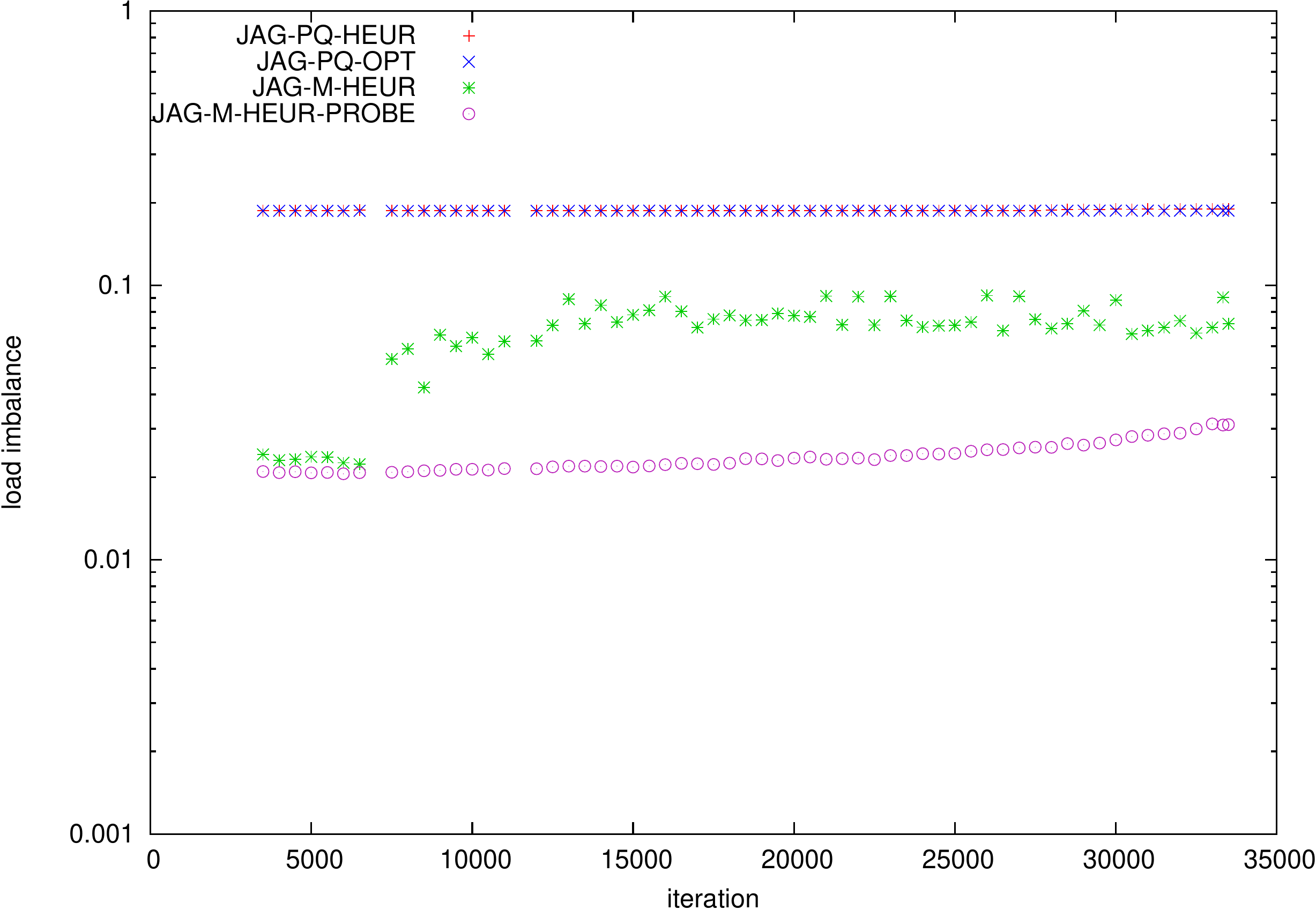}
  \caption{Jagged methods on PIC-MAG with $\nbproc = 6400$.}
  \label{fig:homa3d_proc6400_jagged}
\end{figure}

Figure~\ref{fig:homa3d_proc6400_jagged} presents the load imbalance of
the algorithms with 6,400 processors for the different iterations of
the PIC-MAG application. \pxq-way jagged partitions have a load imbalance of
18\% while the imbalance of the partitions generated by {\tt JAG-M-HEUR}
varies between 2.5\% (at iteration 5,000) and 16\% (at iteration
18,000). {\tt JAG-M-HEUR-PROBE} achieves the best load imbalance of
the heuristics between 2\% and 3\% on all the instances.

In Figure~\ref{fig:homa3d_it30000_jagged}, the optimal \nbproc-way
partition have been computed up to 1,000 processors (on more than
1,000 processors, the runtime of the algorithm becomes
prohibitive). It shows an imbalance of about 1\% at iteration 30,000
of the PIC-MAG application on 1,000 processors. This value is much
smaller than the 6\% imbalance of {\tt JAG-M-HEUR} and {\tt
  JAG-M-HEUR-PROBE}. It indicates that there is room for improvement
for \nbproc-way jagged heuristics. Indeed, the current heuristic uses
$\sqrt{\nbproc}$ parts in the first dimension, while the optimal is
not bounded to that constraint.  Notice that an optimal \nbproc-way
partition with a given number of columns could be computed optimally
using dynamic programming. Figure~\ref{fig:P_study} presents the impact of the
number of stripes on the load imbalance of {\tt JAG-M-HEUR} on a
uniform instance as well as the worst case imbalance of the
\nbproc-way jagged heuristic guaranteed by
Theorem~\ref{th:m-way-guarantee}. It appears clearly that the actual
performance follows the same trend as the worst case performance of
{\tt JAG-M-HEUR}. Therefore, ideally, the number of stripes should be
chosen according to the guarantee of {\tt JAG-M-HEUR}. However, the
parameters of the formula in Theorem~\ref{th:m-way-opt-p} are
difficult to estimate accurately and the variation of the load
imbalance around that value can not be predicted accurately.

\begin{figure}
  \centering
  \includegraphics[width=.8\linewidth]{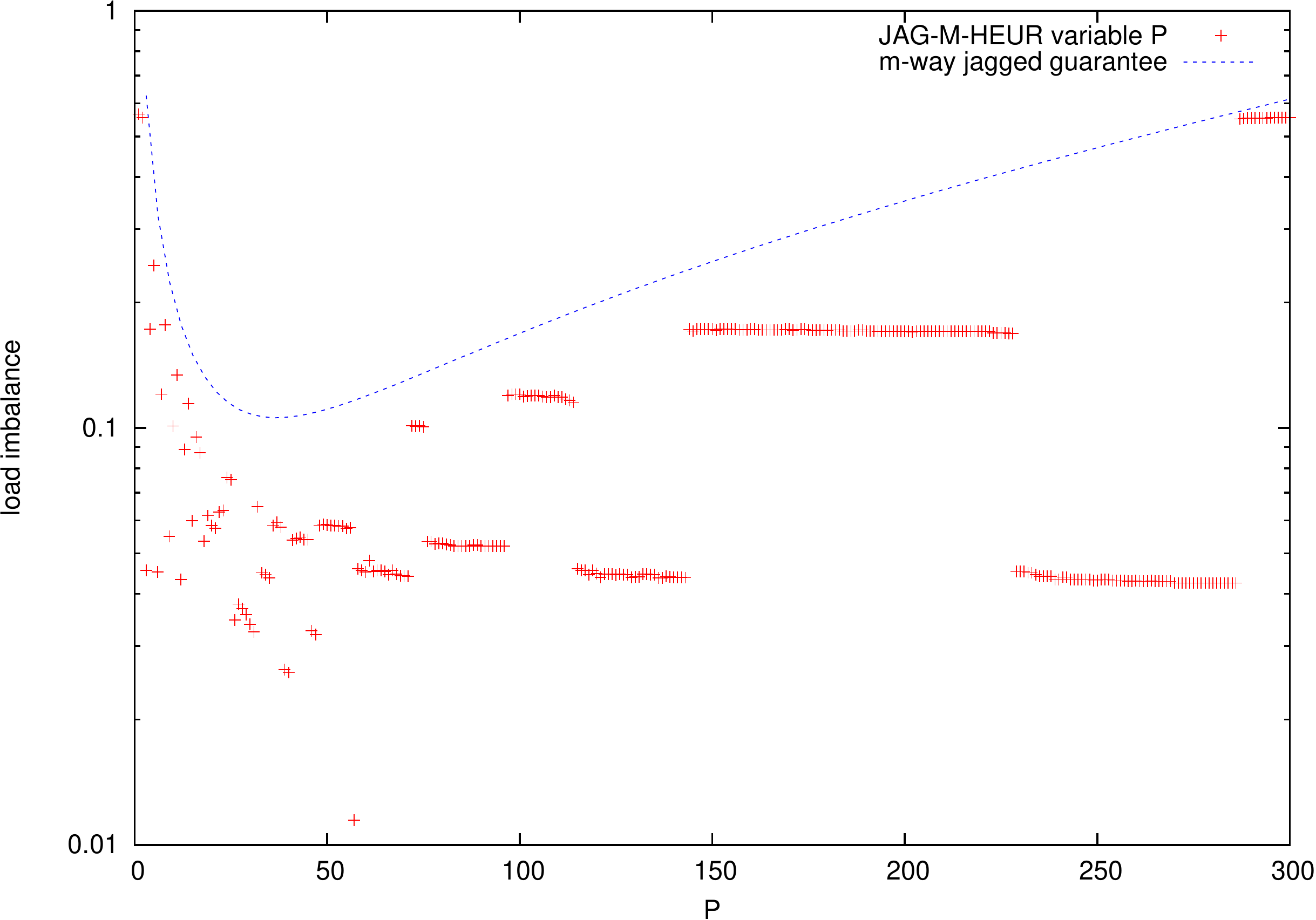}
  \caption{Impact of the number of stripes in {\tt JAG-M-HEUR} on a
    514x514 Uniform instance with $\Delta=1.2$ and $\nbproc = 800$.}
  \label{fig:P_study}
\end{figure}

The load imbalance of {\tt JAG-PQ-HEUR}, {\tt JAG-PQ-OPT}, {\tt
  JAG-M-HEUR} and {\tt JAG-M-HEUR-PROBE} make some waves on
Figure~\ref{fig:homa3d_it30000_jagged} when the number of processors
varies. Those waves are caused by the imbalance of the partitioning in
the main dimension of the jagged partition. Even more, these waves
synchronized with the integral value of $\frac{\sizeX }{
  \sqrt{\nbproc}}$. This behavior is linked to the almost uniformity
of the PIC-MAG dataset. The same phenomena induces the steps in
Figure~\ref{fig:P_study}.

\subsection{Hierarchical Bipartition}

There are four variants of {\tt HIER-RB} depending on the dimension that
will be partitioned in two. In general the load imbalance increases with
the number of processors. The {\tt HIER-RB-LOAD} variant achieves a
slightly smaller load balance than the {\tt HIER-RB-HOR}, {\tt
HIER-RB-VER} and {\tt HIER-RB-DIST} variants. The results are similar on
all the classes of instances and are omitted.

There are also four variants to the {\tt HIER-RELAXED}
algorithm. Figure~\ref{fig:multi_peak_size512_RH} shows the
load imbalance of the four variants when the number of processors
varies on the multi-peak instances of size 512. In general
the load imbalance increases with the number of processors for 
{\tt HIER-RELAXED-LOAD} and {\tt HIER-RELAXED-DIST}. The {\tt
  HIER-RELAXED-LOAD} variant achieves overall the best load
balance. The load imbalance of the {\tt HIER-RELAXED-VER} (and {\tt
  HIER-RELAXED-HOR}) variant improves past 2,000 processors and seems
to converge to the performance of {\tt HIER-RELAXED-LOAD}. The number
of processors where these variants start improving depends on the size
of the load matrix. Before convergence, the obtained load balance is
comparable to the one obtained by {\tt HIER-RELAXED-DIST}. The
diagonal instances with a size of 4,096 presented in
Figure~\ref{fig:diagonal_dist_size4096_recbisectminmax} shows this
behavior.

\begin{figure}
  \centering
  \includegraphics[width=.8\linewidth]{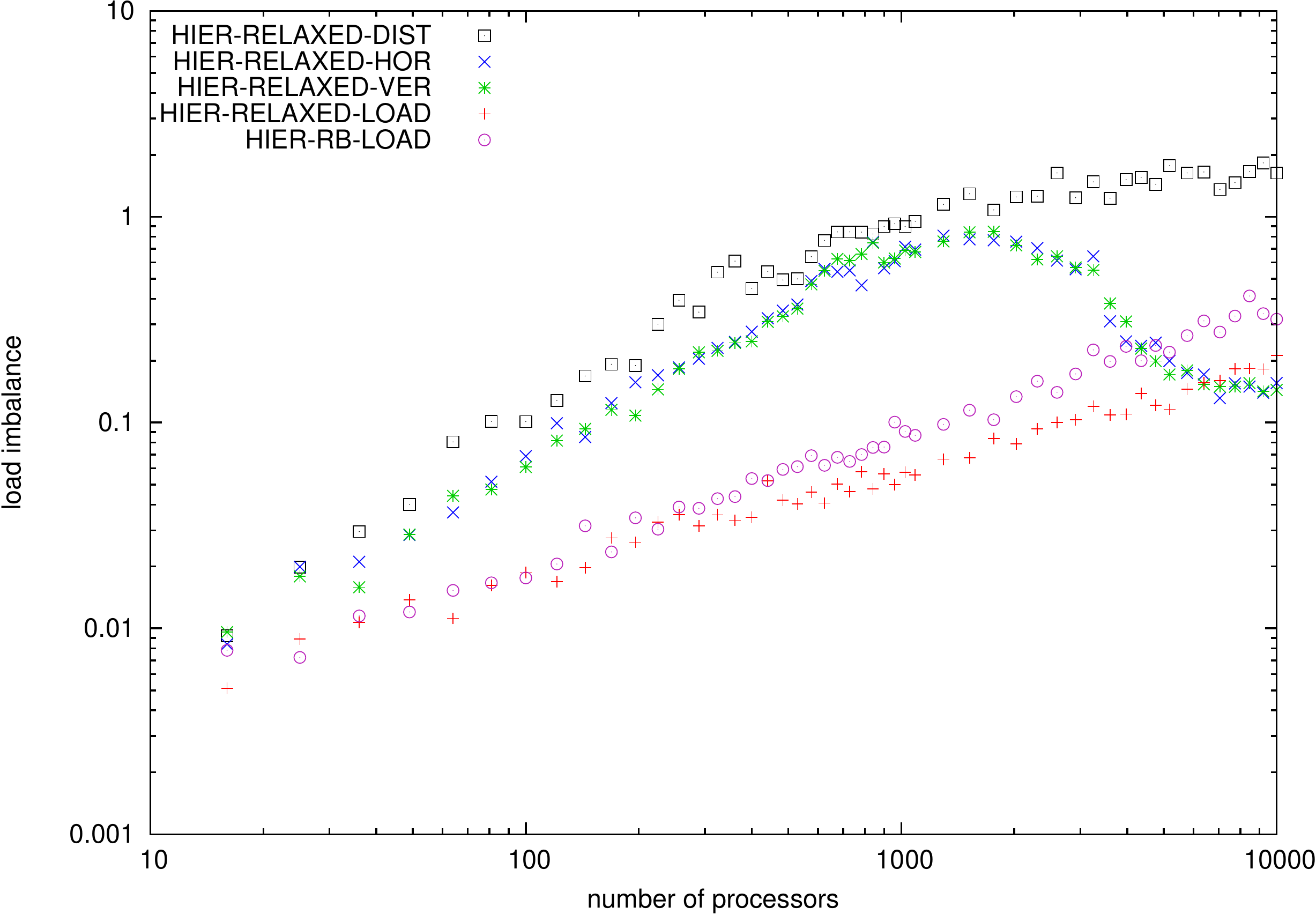}
  \caption{{\tt HIER-RELAXED} on 512x512 Multi-peak.}
  \label{fig:multi_peak_size512_RH}
\end{figure}

\begin{figure}
  \centering
  \includegraphics[width=.8\linewidth]{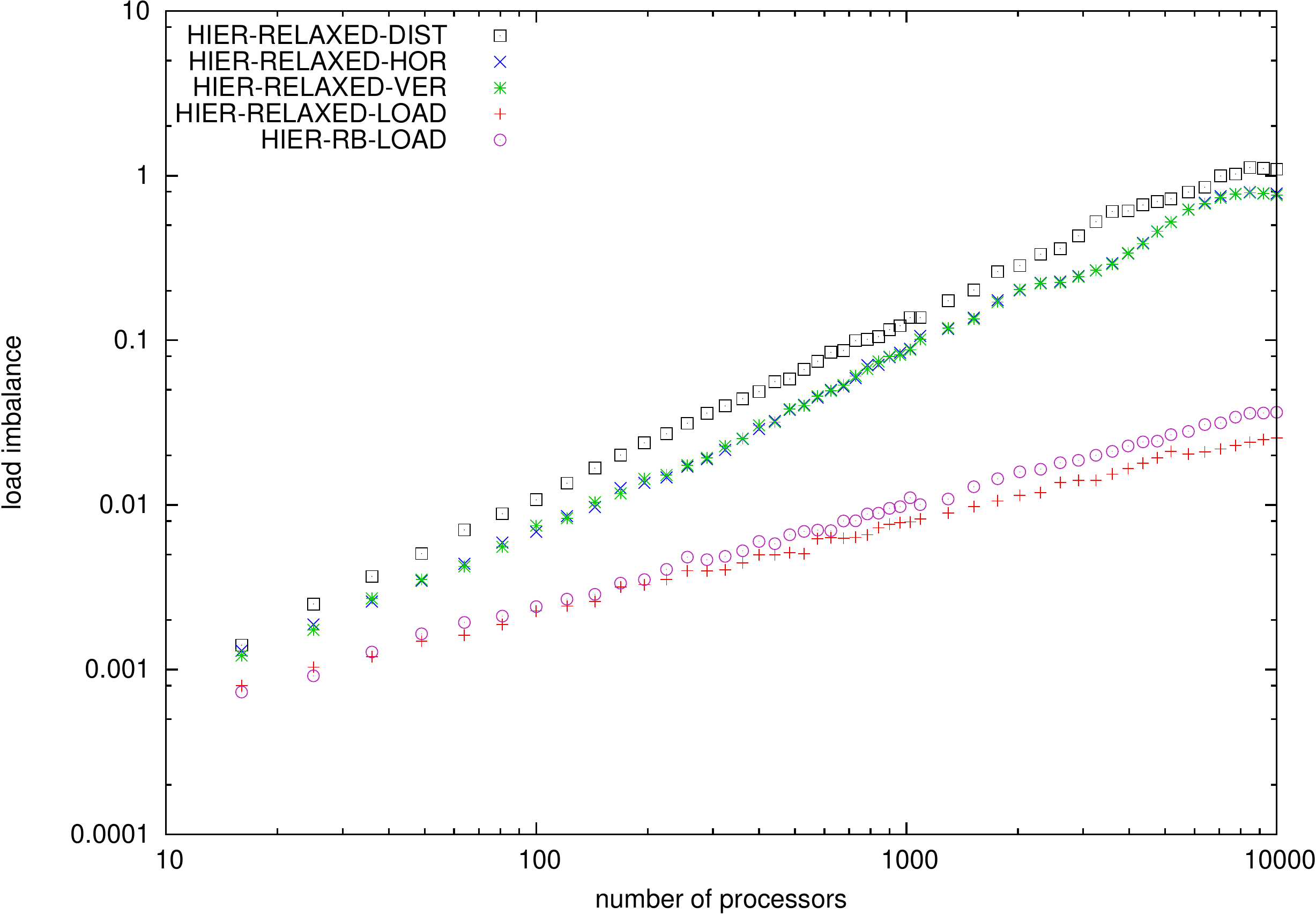}
  \caption{{\tt HIER-RELAXED} on 4096x4096 Diagonal.}
  \label{fig:diagonal_dist_size4096_recbisectminmax}
\end{figure}

Since the load variant of both algorithm leads to the best load
imbalance, we will refer to them as {\tt HIER-RB} and {\tt
HIER-RELAXED}.

We proposed in Section~\ref{sec:hierar_part}, {\tt HIER-OPT}, a
dynamic programming algorithm to compute the optimal hierarchical
bipartition. We did not implement {\tt HIER-OPT} since we expect it to
run in hours even on small instances. However, we derived {\tt
  HIER-RELAXED}, from the dynamic programming
formulation. Figure~\ref{fig:multi_peak_size512_RH}
and~\ref{fig:diagonal_dist_size4096_recbisectminmax} include the
performance of {\tt HIER-RB} and allow to compare it to {\tt
  HIER-RELAXED}. It is clear that {\tt HIER-RELAXED} leads to a better
load balance than {\tt HIER-RB} in these two cases. However, the
performance of {\tt HIER-RELAXED} might be very erratic when the
instance changes slightly. For instance, on
Figure~\ref{fig:homa3d_nbproc400_rebisectc} the performance of {\tt
  HIER-RELAXED} during the execution of the PIC-MAG application is
highly unstable.

\begin{figure}
  \centering
  \includegraphics[width=.8\linewidth]{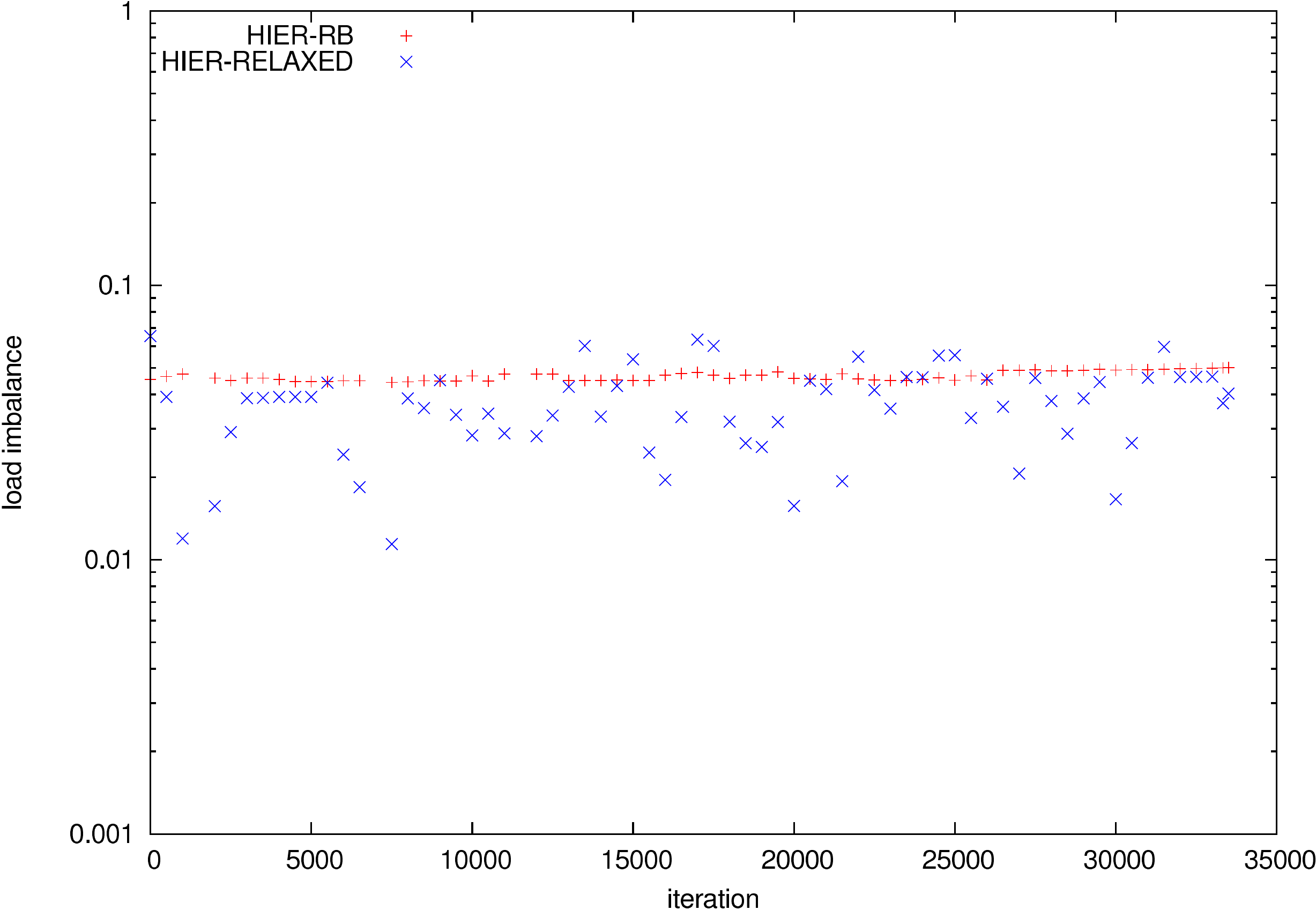}
  \caption{Hierarchical methods on PIC-MAG
    with $\nbproc = 400$.}
  \label{fig:homa3d_nbproc400_rebisectc}
\end{figure}

\subsection{Execution time}

In all optimization problems, the trade-off between the quality of a
solution and the time spent computing it appears. We
present in Figure~\ref{fig:uniform_ratio_rho1_2_time} the execution
time of the different algorithms on 512x512 Uniform instances with
$\Delta=1.2$ when the number of processors varies. The execution times
of the algorithms increase with the number of processors. 

All the heuristics complete in less than one second even on 10,000
processors. The fastest algorithm is obviously {\tt RECT-UNIFORM} since
it outputs trivial partitions. The second fastest algorithm is {\tt
HIER-RB} which computes a partition in 10,000 processors in 18
milliseconds. Then comes the {\tt JAG-PQ-HEUR} and {\tt JAG-M-HEUR}
heuristics which take about 106 milliseconds to compute a solution of
the same number of processors. Notice that the execution of {\tt
JAG-M-HEUR-PROBE} takes about twice longer than {\tt JAG-M-HEUR}. The
running time of {\tt RECT-NICOL} algorithm is more erratic (probably due
to the iterative refinement approach) and it took 448 milliseconds to
compute a partition in 10,000 rectangles. The slowest heuristic is {\tt
HIER-RELAXED} which requires 0.95 seconds of computation to compute a
solution for 10,000 processors.

Two algorithms are available to compute the optimal \pxq-way jagged
partition. Despite the various optimizations implemented in the dynamic
programming algorithm, {\tt JAG-PQ-OPT-DP} is about one order of
magnitude slower than {\tt JAG-PQ-OPT-NICOL}. {\tt JAG-PQ-OPT-DP} takes
63 seconds to compute the solution on 10,000 processors whereas {\tt
JAG-PQ-OPT-NICOL} only needs 9.6 seconds. Notice that using heuristic
algorithm {\tt JAG-PQ-HEUR} is two order of magnitude faster than {\tt
JAG-PQ-OPT-NICOL}, the fastest known optimal \pxq-way jagged algorithm.

The computation time of {\tt JAG-M-OPT} is not reported on the chart. We
never run this algorithm on a large number of processors since it already
took 31 minutes to compute a solution for 961 processors. The results on
different classes of instances are not reported, but show the same
trends. Experiments with larger load matrices show an increase in the
execution time of the algorithm. Running the algorithms on matrices of
size 8,192x8,192 basically increases the running times by an order of
magnitude.

\begin{figure}
  \centering
  \includegraphics[width=.8\linewidth]{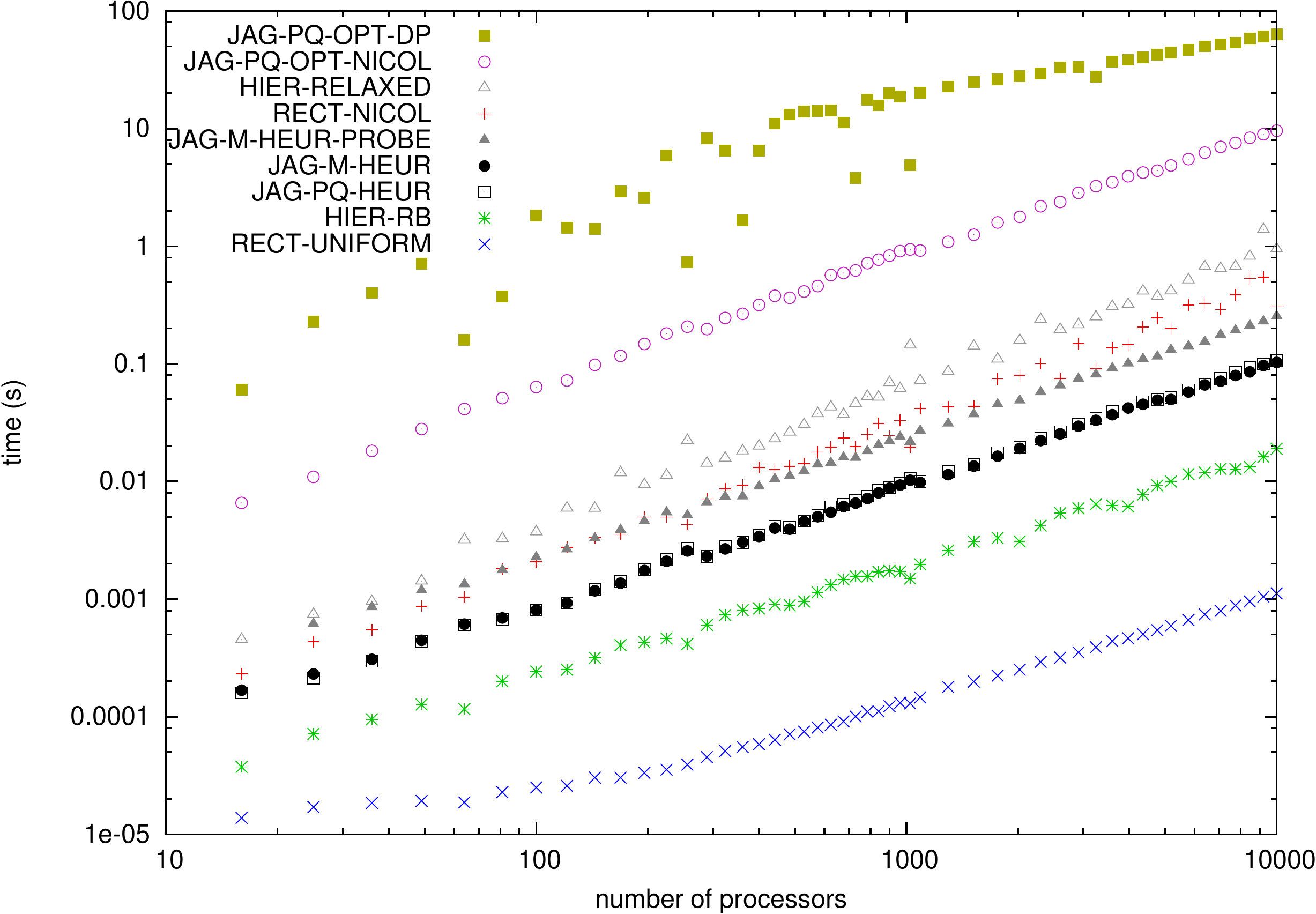}
  \caption{Runtime on 512x512 Uniform with $\Delta=1.2$.}
  \label{fig:uniform_ratio_rho1_2_time}
\end{figure}

Loading the data and computing the prefix sum array is required by all
two dimensional algorithms. Hence, the time taken by these operations
is not included in the presented timing results. For reference, it is
about 40 milliseconds on a 512x512 matrix.

\subsection{Which algorithm to choose?}

The main question remains. Which algorithm should be chosen to
optimize an application's performance?

From the algorithm we presented, we showed that \nbproc-way jagged
partitioning techniques provide better solutions than an optimal \pxq-way
jagged partition. It is therefore better than rectilinear partitions
as well. The computation of an optimal \nbproc-way jagged partition is
too slow to be used in a real system. It remains to decide between
{\tt JAG-M-HEUR-PROBE}, {\tt HIER-RB} and {\tt HIER-RELAXED}. As a
point of reference, the results presented in this section also include
the result of algorithm generating rectilinear partitioning, namely,
{\tt RECT-UNIFORM} and {\tt RECT-NICOL}.

\begin{figure}
  \centering
  \includegraphics[width=.8\linewidth]{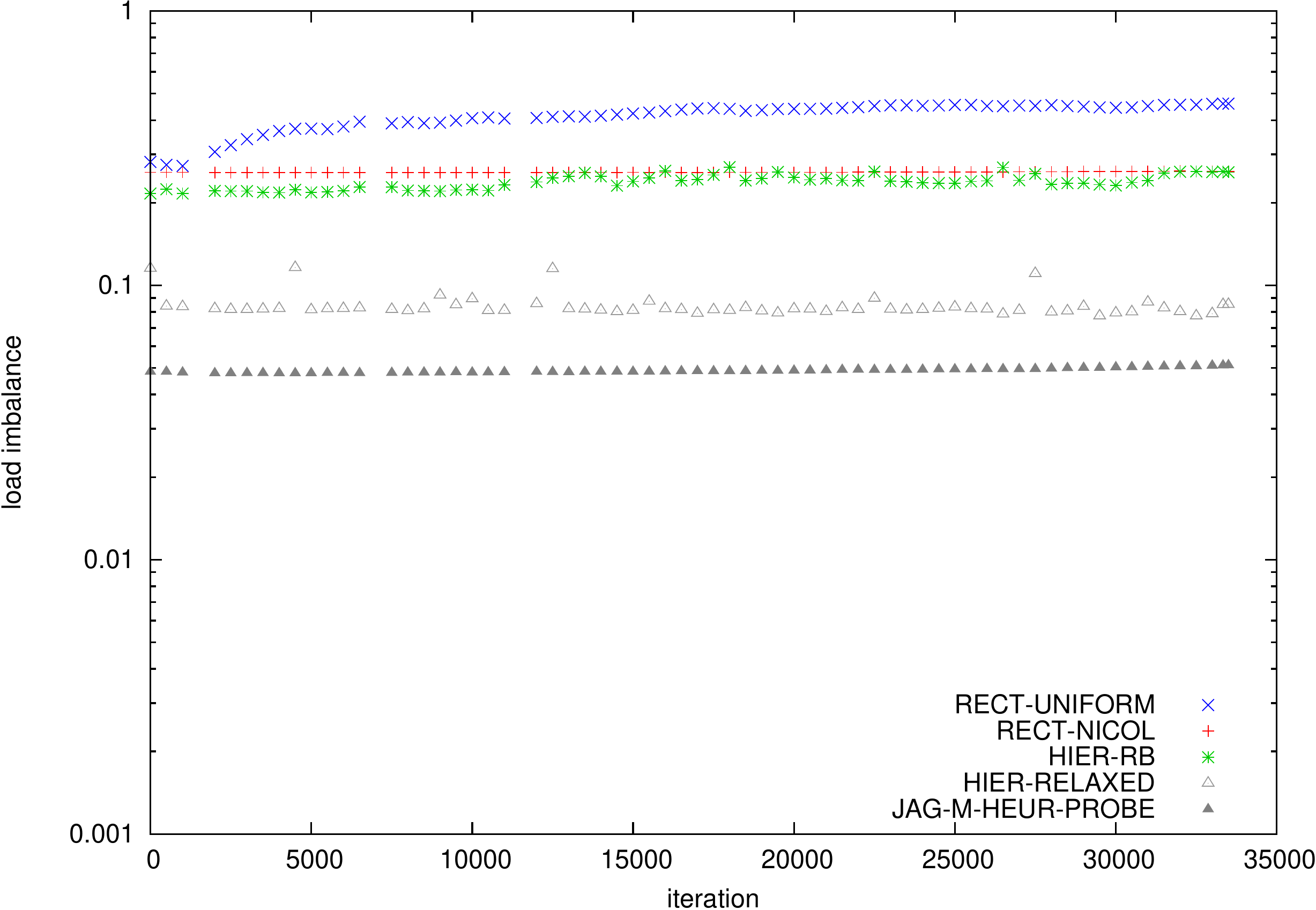}
  \caption{Main heuristics on PIC-MAG with $\nbproc = 9216$.}
  \label{fig:homa3d_proc9216_recap}
\end{figure}

Figure~\ref{fig:homa3d_proc9216_recap} shows the performance of the
PIC-MAG application on 9,216 processors. The {\tt RECT-UNIFORM}
partitioning algorithm is given as a reference. It achieves a load
imbalance that grows from 30\% to 45\%. {\tt RECT-NICOL} reaches a
constant 28\% imbalance over time. {\tt HIER-RB} is usually slightly
better and achieves a load imbalance that varies between 20\% and
30\%. {\tt HIER-RELAXED} achieves most of the time a much better load
imbalance, rarely over 10\% and typically between 8\% and 9\%. {\tt
  JAG-M-HEUR-PROBE} outperforms all the other algorithms by providing
a constant 5\% load imbalance.

Figure~\ref{fig:homa3d_it20000_recap} shows the performance of the
algorithms while varying the number of processors at iteration
20,000. The conclusions on {\tt RECT-UNIFORM}, {\tt RECT-NICOL} and
{\tt HIER-RB} stand. Depending on the number of processors, the
performance of {\tt JAG-M-HEUR-PROBE} varies and in general {\tt
  HIER-RELAXED} leads to the best performance, in this test.

\begin{figure}
  \centering
  \includegraphics[width=.8\linewidth]{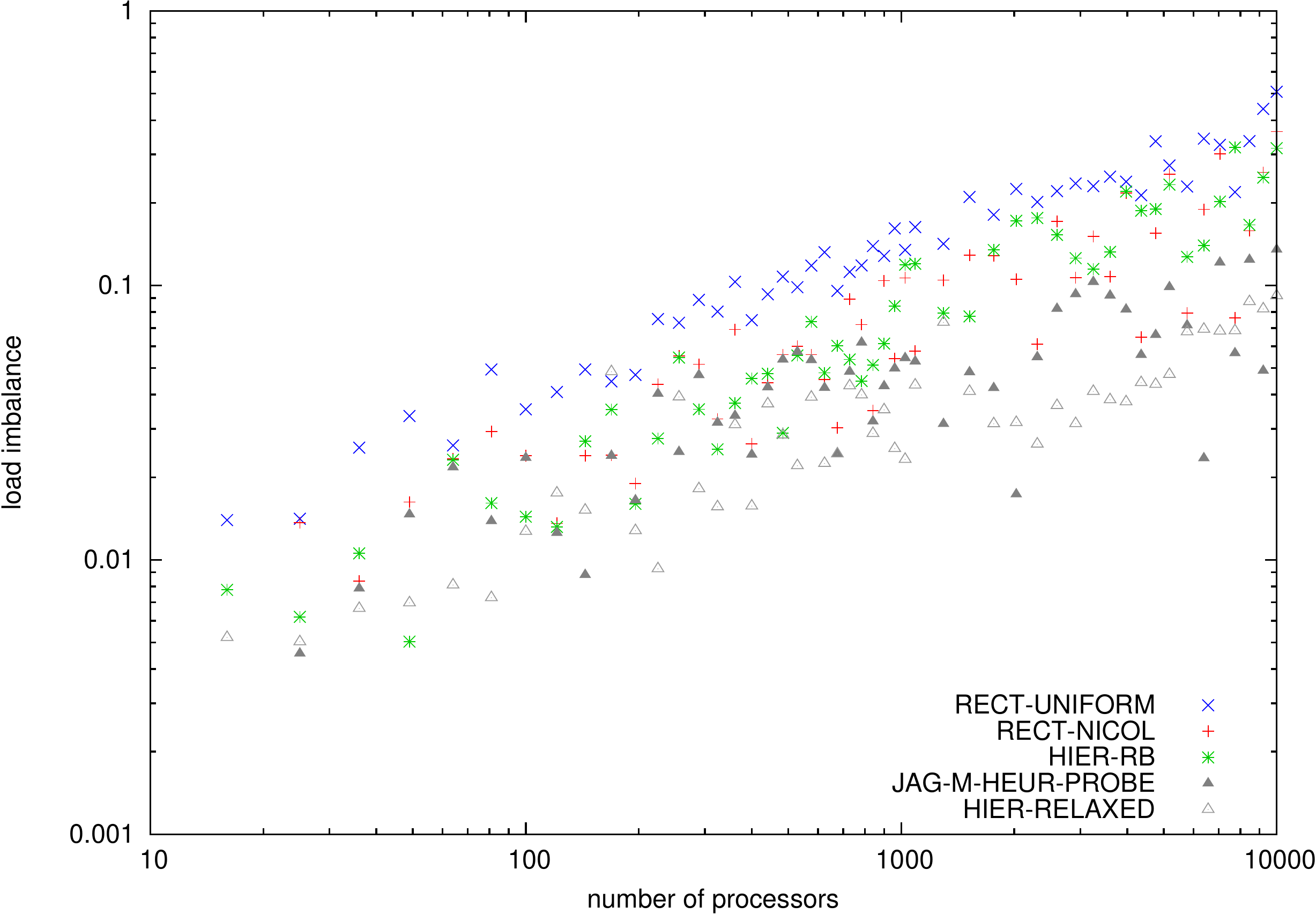}
  \caption{Main heuristics on PIC-MAG iter=20,000.}
  \label{fig:homa3d_it20000_recap}
\end{figure}

Figure~\ref{fig:SLAC6.0M_yz_recap} presents the performance of the
algorithms on the mesh based instance SLAC. Due to the sparsity of
the instance, most algorithms get a high load imbalance.  Only the
hierarchical partitioning algorithms manage to keep the imbalance low and
{\tt HIER-RELAXED} gets a lower imbalance than {\tt HIER-RB}.

\begin{figure}
  \centering
  \includegraphics[width=.8\linewidth]{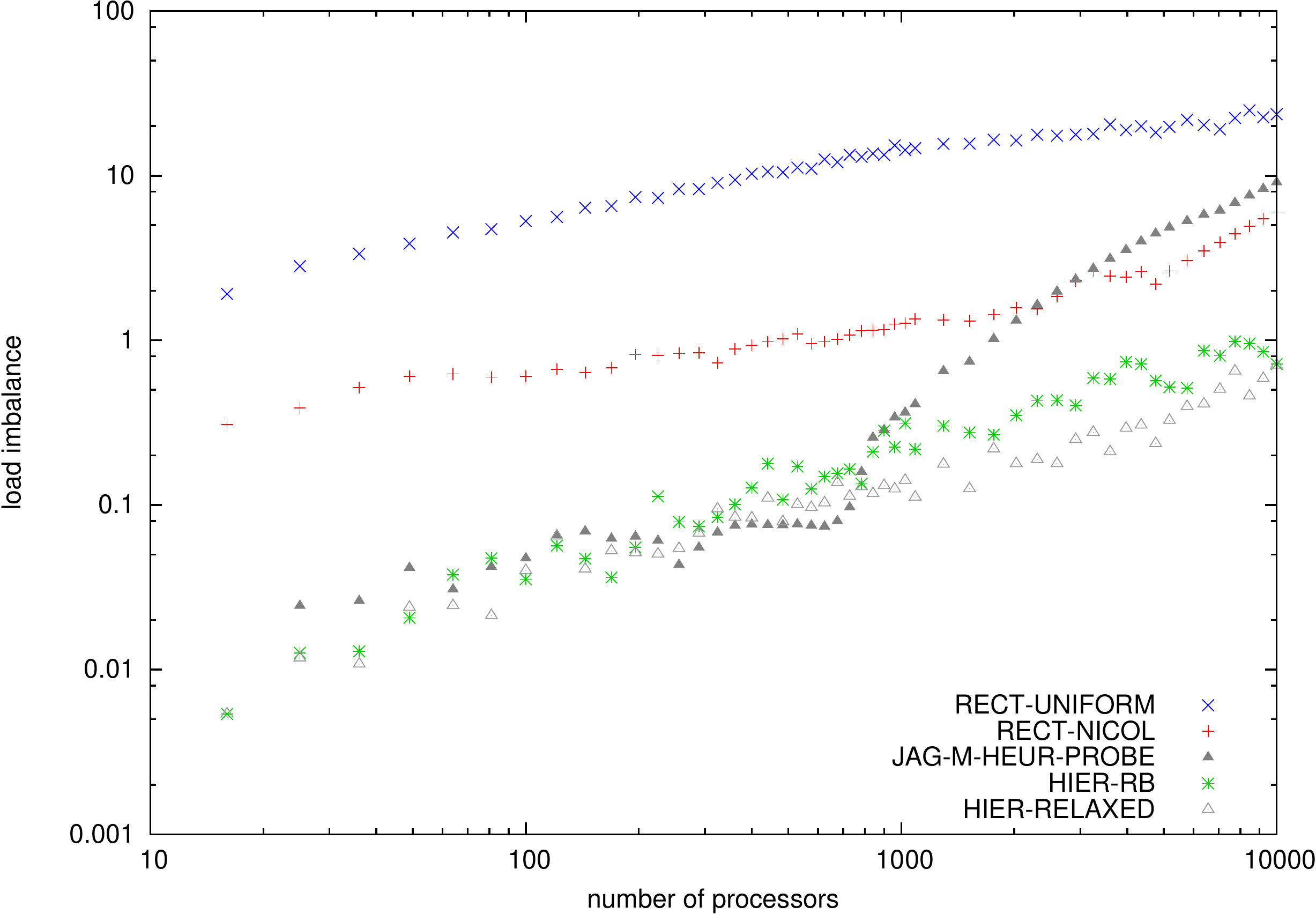}
  \caption{Main heuristics on SLAC.}
  \label{fig:SLAC6.0M_yz_recap}
\end{figure}

The results indicate that as it stands, the algorithms {\tt
  HIER-RELAXED} and {\tt JAG-M-HEUR-PROBE}, we proposed, are the one
to choose to get a good load balance. However, we believe a developer
should be cautious when using {\tt HIER-RELAXED} because of the
erratic behavior it showed in some experiments (see
Figure~\ref{fig:homa3d_nbproc400_rebisectc}) and because of its
not-that-low running time (up to one second on 10,000 processors
according to Figure~\ref{fig:uniform_ratio_rho1_2_time}). {\tt
  JAG-M-HEUR-PROBE} seems much a more stable heuristic. The bad load
balance it presents on Figure~\ref{fig:homa3d_it20000_recap} is due to
a badly chosen number of partitions in the first dimension.

\section{Hybrid partitioning scheme}
\label{sec:hybrid}

The previous sections show that we have on one hand, heuristics that
are good and fast, and on the other hand, optimal algorithms which are
even better but to slow to be used in most practical cases. This
section presents some engineering techniques one can use to obtain
better results than using only the heuristics while keeping the runtime of
the algorithms reasonable.

Provided, in general the maximum load of a partition is given by the most loaded
rectangle and not by the general structure of the partition, one idea
is to use the optimal algorithm to be locally efficient and leave the
general structure to a faster algorithm. We introduce the class of
{\tt HYBRID} algorithms which construct a solution in two phases. A
first algorithm will be used to partition the matrix $\loadmatrix$ in
$P$ parts. Then the parts will be independently partitioned with a
second algorithm to obtain a solution in $\nbproc$ parts. This section
investigates the hybrid algorithms and try to answer the following
questions: Which algorithms should be used at phase 1 and phase 2? In
how many parts the matrix should be divided in the first phase (i.e.,
what should $P$ be)?
How to allocate the $\nbproc$ processors between the $P$ parts? And
most importantly, is there any advantage in using hybrid algorithms?

Between the two phases, it is difficult to know how to allocate the
$\nbproc$ processors to the $P$ parts without doing a deep search. We
choose to allocate the parts proportionally according to the rule used
in {\tt JAG-M-HEUR}, i.e., each rectangle $r$ will first be allocated
$Q_r = \ceil{\frac{\load(r)}{\load(\loadmatrix)}(\nbproc-P)}$ parts. The
remaining processors are distributed greedily.

We conducted experiments using different PIC-MAG instances. All the
values of $P$ were tried between $2$ and $\frac{\nbproc}{2}$.
We will denote the {\tt HYBRID} algorithm using {\tt ALGO1} for phase
1 and {\tt ALGO2} for phase 2 as {\tt HYBRID(ALGO1/ALGO2)}.

The first round of experiments mainly showed three observations. (No
results are shown since similar results will be presented later.)
First, {\tt HYBRID} is too slow to use {\tt JAG-M-OPT} at the second
phase for studying the performance (e.g., partitioning
PIC-MAG at iteration 5000 on $1024$ processors using $P=17$ takes $78$
seconds). Second, the performance shows "waves" when $P$ varies which
are correlated with the values of
$\ceil{\frac{\nbproc-P}{P}}$. Finally {\tt HYBRID} can obtain load
imbalances better than {\tt JAG-M-HEUR} and {\tt HIER-RELAXED} on some
configuration confirming that {\tt HYBRID} might be useful.

To make the algorithm faster, we introduce the notion of {\em fast} and
{\em slow} algorithms at phase 2. The {\em fast} algorithm is first run
on each part and the parts are sorted according to their maximum
load. The {\em slow} algorithm is run on the part of higher maximum
load. If the solution returned by the {\em slow} algorithm improves the
maximum load of that part, the solution is kept and the parts are sorted
again. Otherwise, the algorithm terminates. This modification increased
the speed of the algorithm up to an order of magnitude. (Using {\tt
JAG-M-HEUR-PROBE} as the {\em fast} algorithm in phase 2 allows to run
PIC-MAG at iteration 5000 on $1024$ processors using $P=17$ in $38$
seconds, halving the computation time.) Detailed
timing on PIC-MAG at iteration 5000 using $1024$ processors can be found in
Figure~\ref{fig:hybrid-runtime-fast}. The {\tt
HYBRID(JAG-M-HEUR/JAG-M-OPT)} curve is the original implementation of
{\tt HYBRID} using {\tt JAG-M-HEUR} at phase 1 and {\tt JAG-M-OPT} at
phase 2. The {\tt HYBRID-F(JAG-M-HEUR/JAG-M-OPT)} curve presents the
timing obtained with the use of {\em fast} and {\em slow} algorithm. The
{\tt HYBRID} algorithm using {\tt JAG-M-OPT} at both phase is given as a
point of reference. All the implementation used {\tt JAG-M-HEUR-PROBE} as the
{\em fast} algorithm. Figure~\ref{fig:hybrid-runtime-fast} shows that
using {\em fast} and {\em slow} algorithms makes the computation about
one order of magnitude faster. Using {\tt JAG-M-OPT} at phase 1 is
typically orders of magnitude slower than using another algorithm.

\begin{figure}[htb]
  \centering
  \includegraphics[width=.8\linewidth]{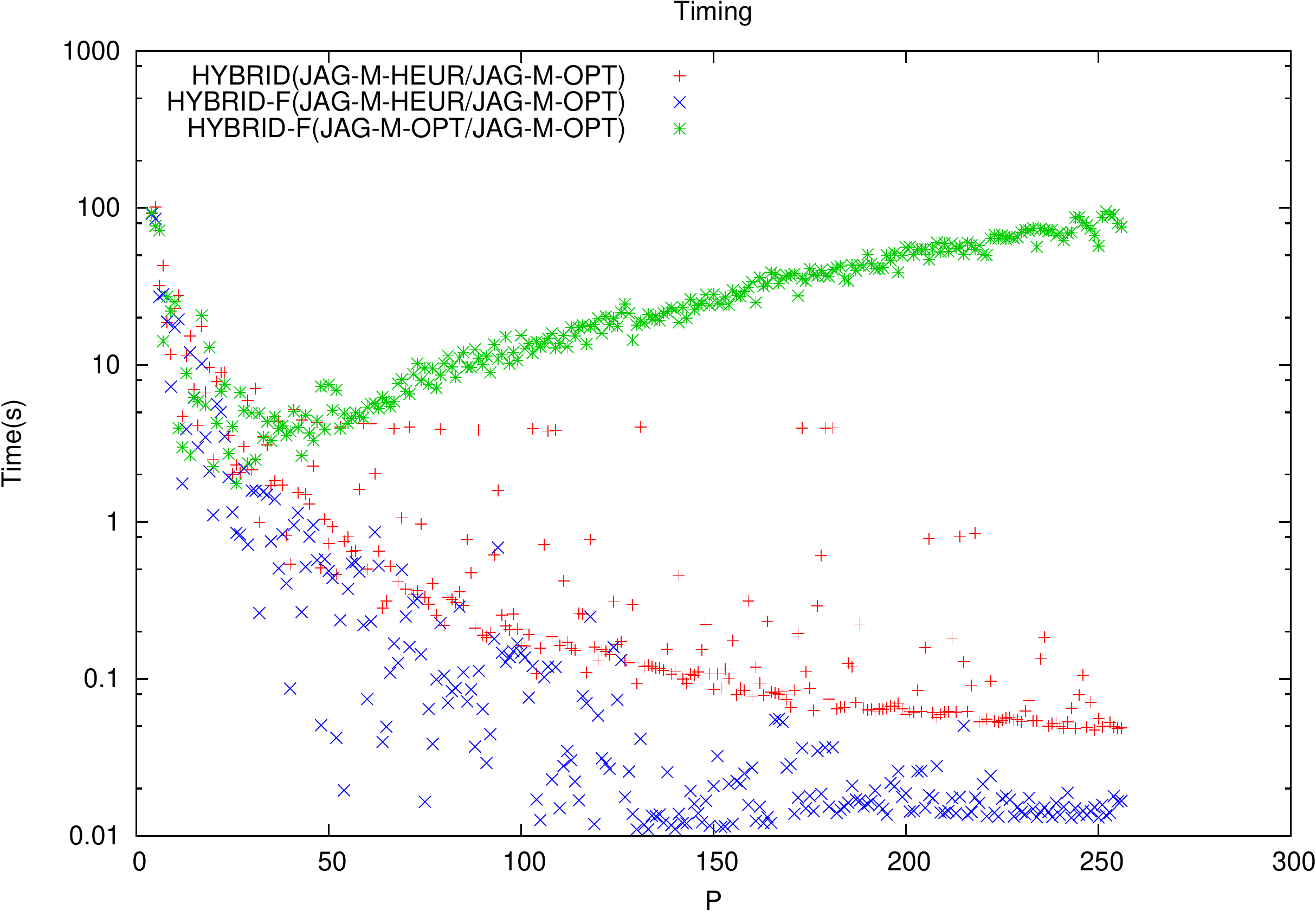}
  \caption{Runtime of {\tt HYBRID} methods on PIC-MAG iter=5000 with $\nbproc=512$.}
  \label{fig:hybrid-runtime-fast}
\end{figure}

This improvement allows to run more complete and detailed
experiments. In particular, using {\tt JAG-M-OPT} at phase 2 runs
quickly. This allow us to study the performance of {\tt HYBRID} using
an algorithm that get good load balance at phase 2. The load imbalance
obtained on PIC-MAG 5000 on $512$ are shown in
Figure~\ref{fig:hybrid-optl2}. Two {\tt HYBRID} variants that use {\tt
  JAG-M-OPT} or {\tt HIER-RELAXED} at phase 1 are
presented. For reference, three horizontal lines present the
performance obtained by {\tt JAG-M-HEUR}, {\tt HIER-RELAXED}
and {\tt JAG-M-OPT} on that instance. A first remark is that a
large number of configurations lead to load imbalances better than {\tt
  JAG-M-HEUR}. A significant number of them get load imbalances
better than {\tt HIER-RELAXED} and sometimes comparable to the
performance of {\tt JAG-M-OPT}. Then, the load imbalance
significantly varies with $P$: it is decreasing by interval which
happen to be synchronized with the values of
$\ceil{\frac{\nbproc-P}{P}}$. Finally, the load imbalance is better
when the values of $P$ are low. This behavior was predictable since
the lower $P$ is the more global the optimization is. However, one
should notice that some low load imbalances are found with high values
of $P$.

\begin{figure}[htb]
  \centering
  \includegraphics[width=.8\linewidth]{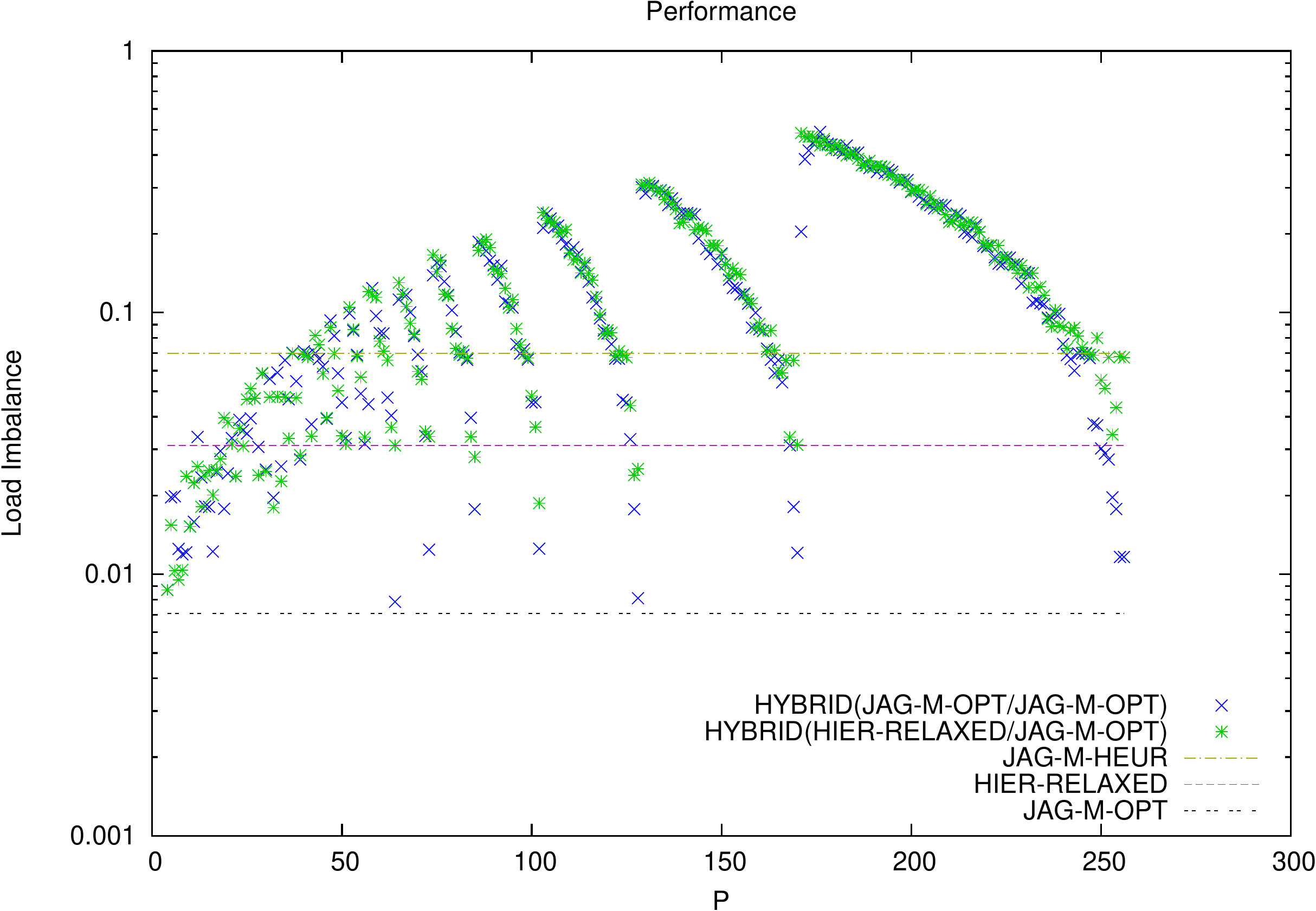}
  \caption{Using {\tt JAG-M-OPT} at phase 2 on PIC-MAG iter=5000 with $\nbproc = 512$}
  \label{fig:hybrid-optl2}
\end{figure}

Good load imbalance could obviously be obtained by trying every single
value of $P$. However, such a procedure is likely to take a lot of
time. Provided the phase 2 algorithm takes a large part of the
computation time, it will be interesting to predict the performance of
the second phase without having to run it. We define the expected load
imbalance as the load imbalance that would be obtained provided the
second phase balances the load optimally, i.e., $eLI =
\max_r \frac{\load(r)}{Q_r}$. Figure~\ref{fig:hybrid-correlation}
presents for each solution its expected load imbalance and the load
imbalance obtained once the second phase is run for different values
of $P$. The solutions are presented for two hybrid variants, one using
{\tt JAG-M-HEUR} at phase 2 and the other one using {\tt
  JAG-M-OPT}. For similar expected load imbalance, the load
imbalance obtained using {\tt JAG-M-HEUR} are spread over an
order of magnitude. However, the load imbalance obtained by {\tt
  JAG-M-OPT} are much more focused. The expected load imbalance
and obtained load imbalance are well correlated when {\tt
  JAG-M-OPT} is used at phase 2.

\begin{figure}[htb]
  \centering
  \includegraphics[width=.8\linewidth]{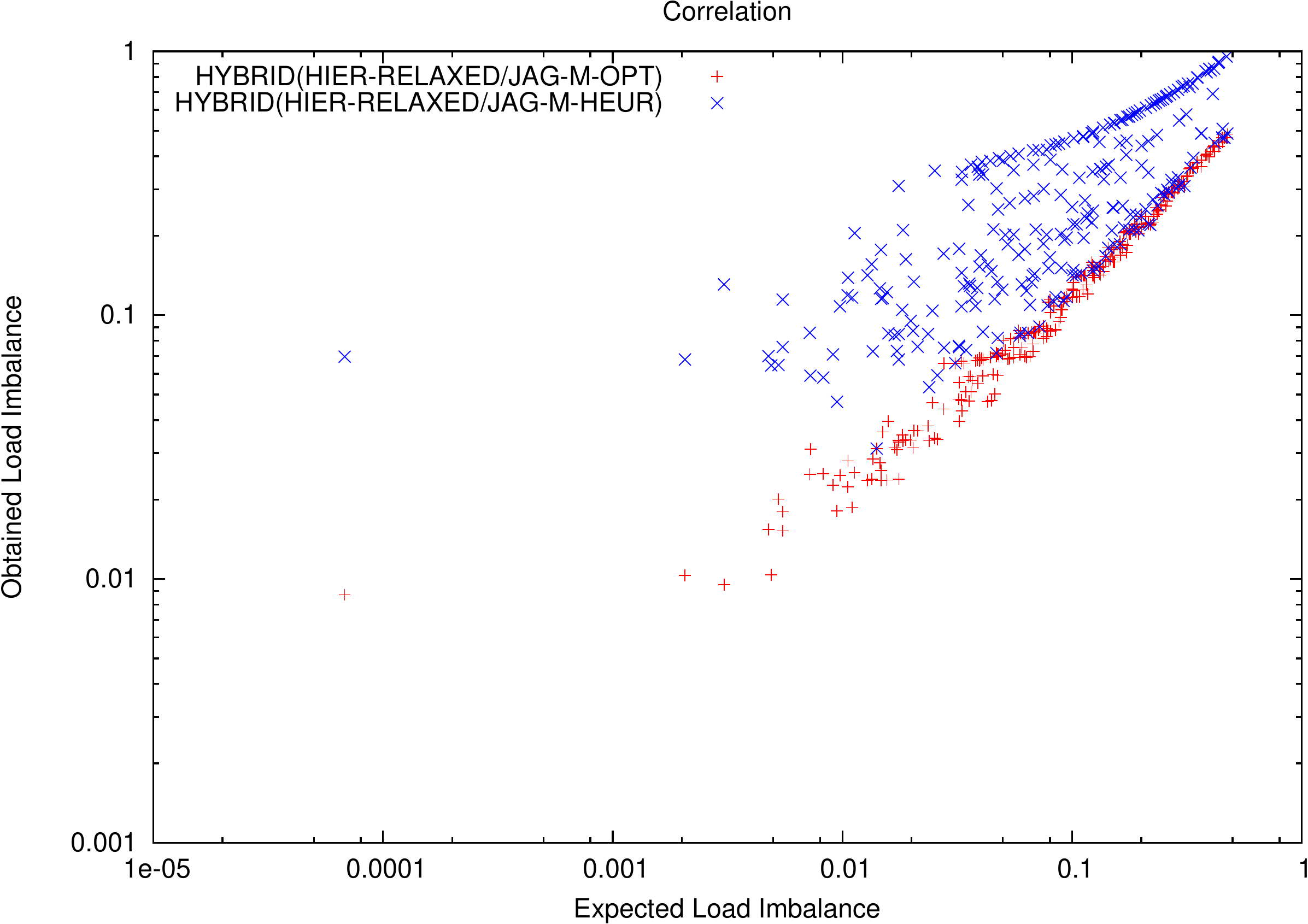}
  \caption{Correlation between expected load imbalance at the end of
  phase 1 and the obtained load imbalance on PIC-MAG iter=5000 with
  $\nbproc = 512$} \label{fig:hybrid-correlation}
\end{figure}

The previous experiments show two things. The actual performance are
correlated with expected performance at the end of phase 1 if {\tt
JAG-M-OPT} is used in phase 2. The load imbalance decreases in an
interval of values of $P$ synchronized with the values of
$\ceil{\frac{\nbproc-P}{P}}$. Therefore, we propose to enumerate the
values of $P$ at the end of such intervals. For each of these value, the
phase 1 algorithm is used and the expected load imbalance is
computed. The phase 2 is only applied on the value of $P$ leading to the
best expected load imbalance. Obviously the best expected load imbalance
will be given by the non-hybrid case $P=1$, but it will lead to a high
runtime. The tradeoff between the runtime of the algorithm and the
quality of the solution should be left to the user by specifying a
minimal $P$.

Figure \ref{fig:hybrid-PICMAG-it10000} presents the load imbalance
obtained on the PIC-MAG datasets at iteration 10000 using
$\sqrt{\nbproc}$ as minimal $P$. The {\tt HYBRID} algorithm obtains a
load balance usually better than {\tt JAG-M-HEUR-PROBE} and often
better than {\tt HIER-RELAXED}. The algorithm leading to the best load
imbalance seems to depend on the number of processors. For instance,
Figure~\ref{fig:hybrid-PICMAG-m7744} shows the load imbalance of the
algorithms on 7744 processors. {\tt JAG-M-HEUR-PROBE} leads constantly
to better results than {\tt HIER-RELAXED}. And {\tt HYBRID} typically
improve both by a few percents.

\begin{figure}[htb]
  \centering
  \includegraphics[width=.8\linewidth]{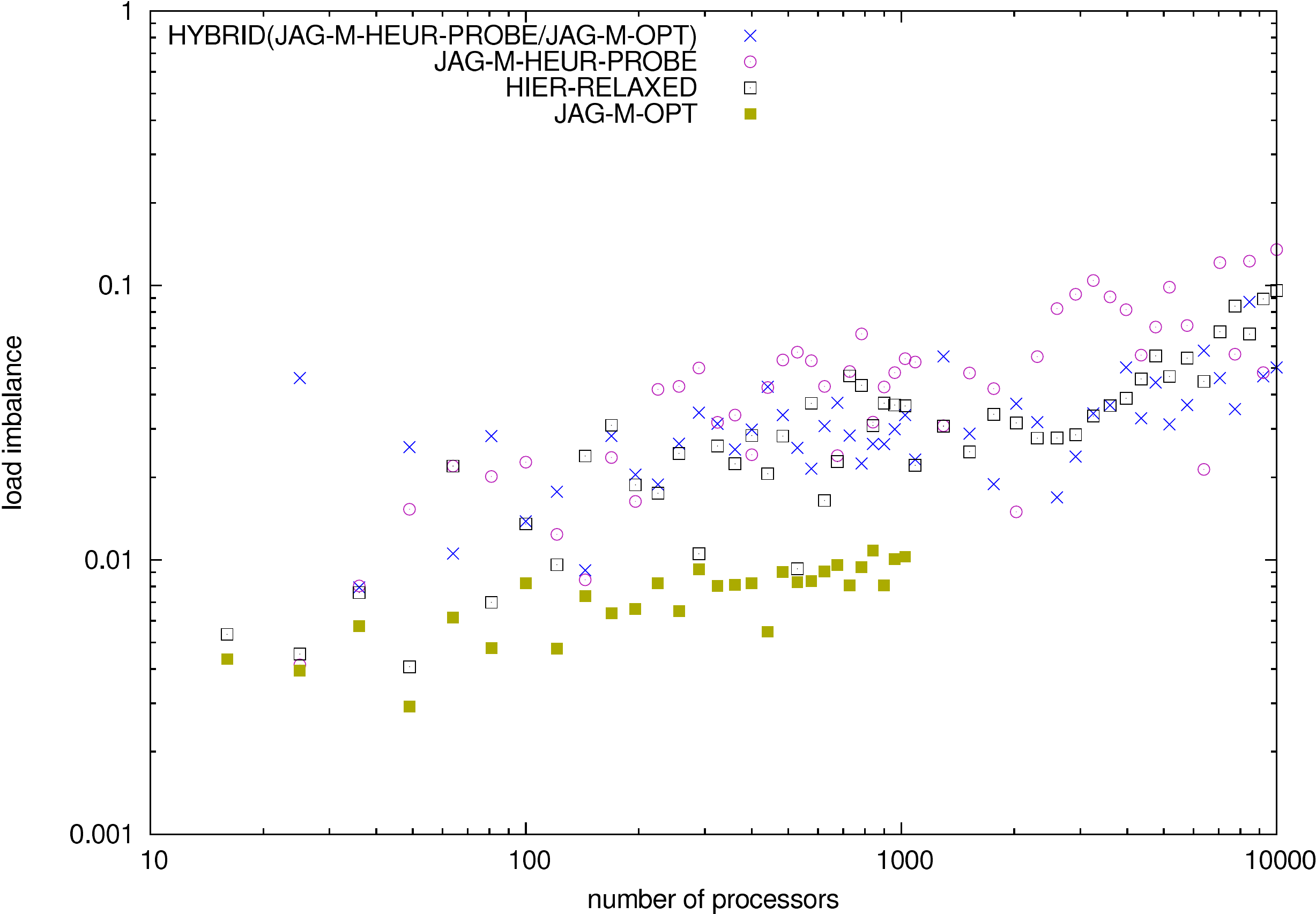}
  \caption{{\tt HYBRID} algorithm on PIC-MAG iter=10000}
  \label{fig:hybrid-PICMAG-it10000}
\end{figure}

\begin{figure}[htb]
  \centering
  \includegraphics[width=.8\linewidth]{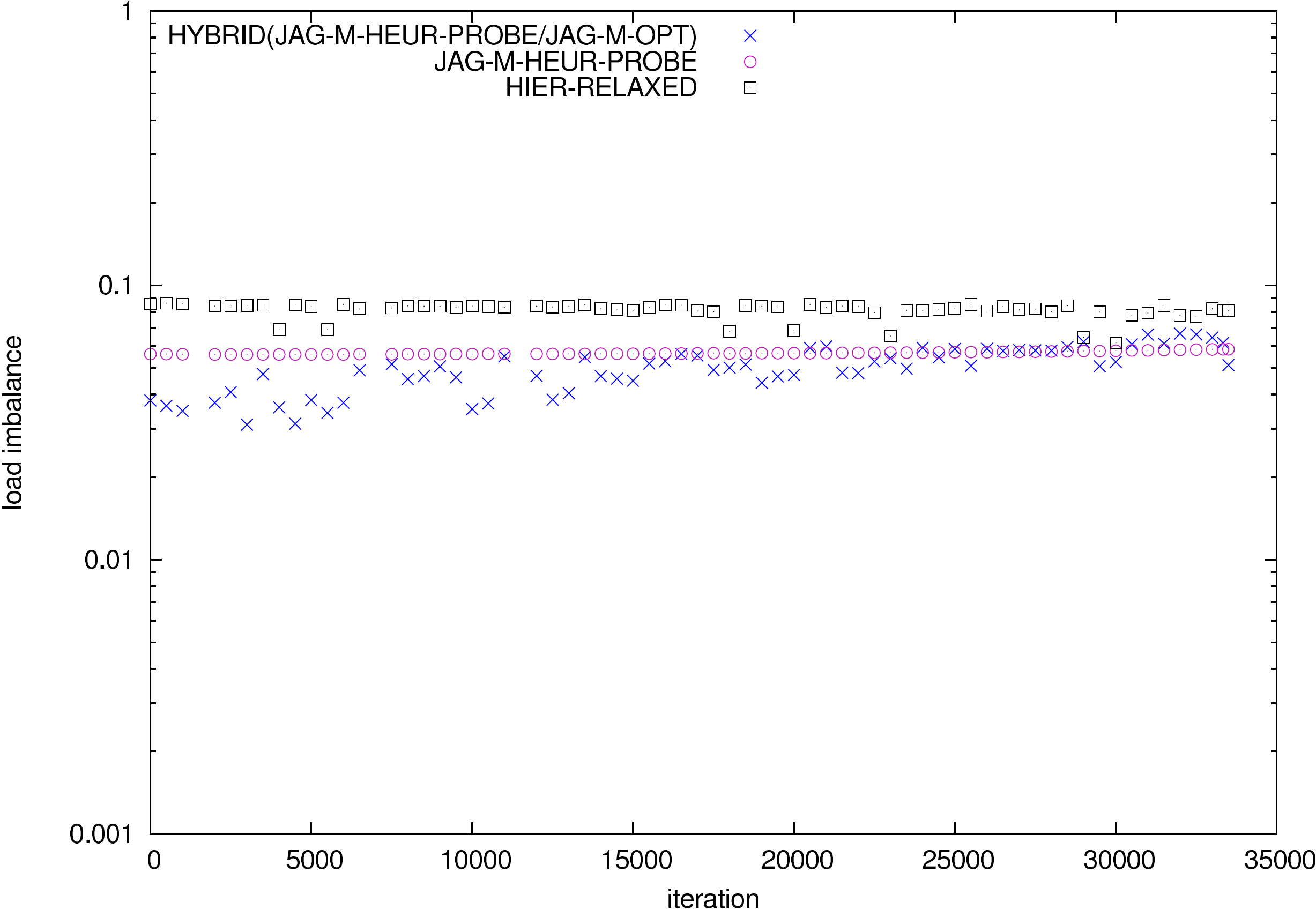}
  \caption{{\tt HYBRID} algorithm on PIC-MAG on 7744 processors}
  \label{fig:hybrid-PICMAG-m7744}
\end{figure}

However, on 6400 processors (Figure~\ref{fig:hybrid-PICMAG-m6400}),
{\tt HYBRID} almost constantly improves the result of {\tt
  HIER-RELAXED} by a few percents but does not achieve better load
imbalance than {\tt
  JAG-M-HEUR-PROBE}. Figure~\ref{fig:homa3d_proc6400_jagged} showed that
{\tt JAG-M-HEUR} is significantly outperformed by {\tt
  JAG-M-HEUR-PROBE} in that configuration. Recall that the main
difference between these heuristics is that the former distributes the
processors among the stripes only based on the load of each stripe
while the latter use the minimum number of processors per stripe to
obtain the minimum load balance. The same idea could be applied to
{\tt HYBRID} algorithms looking for the minimum number of processors to
allocate to each part without degrading the load imbalance and use
these processors on the parts that lead to the maximum load. This
modification will improve the load imbalance but will also increase
the running time significantly.

\begin{figure}[htb]
  \centering
  \includegraphics[width=.8\linewidth]{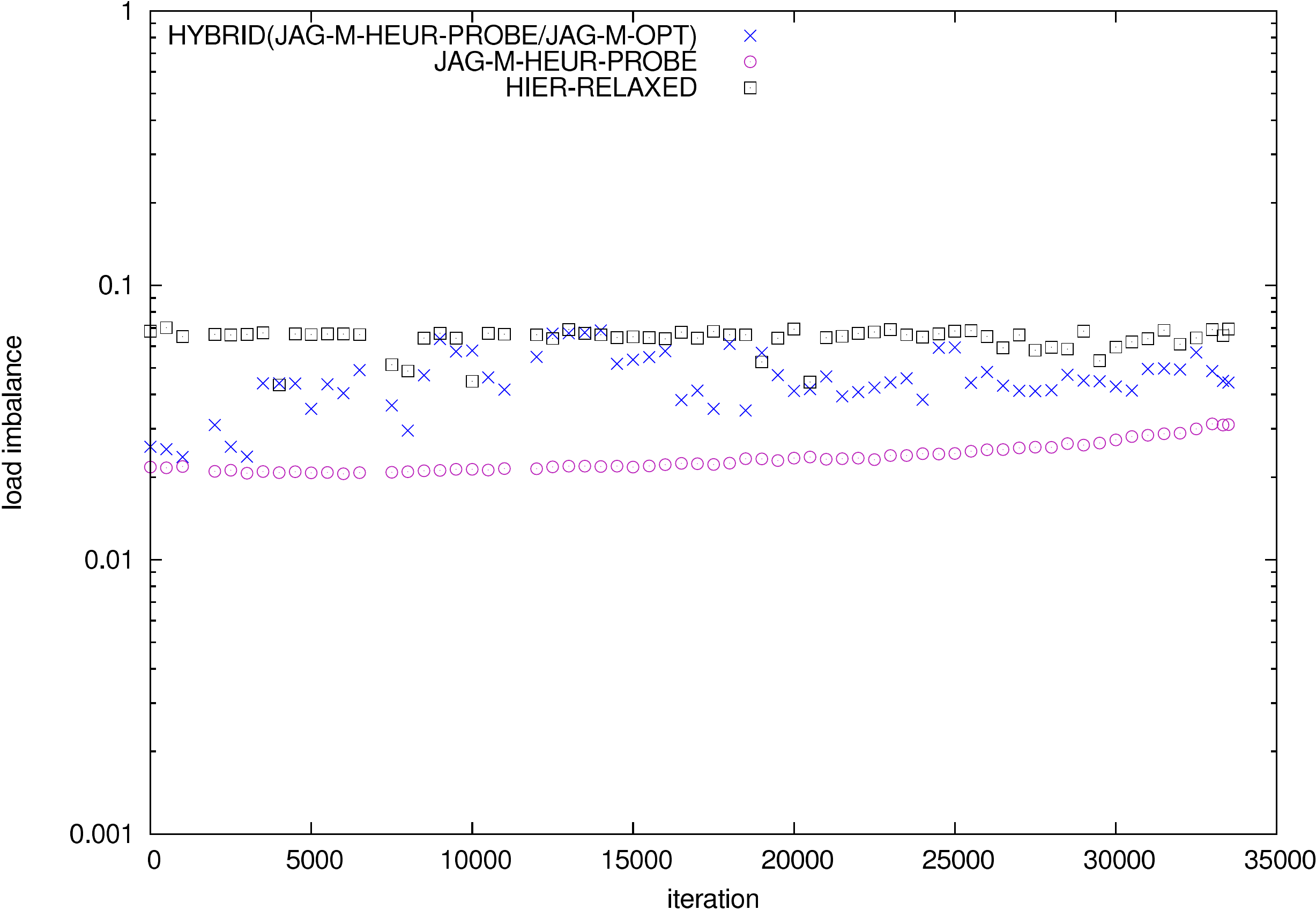}
  \vspace*{-2ex}
  \caption{{\tt HYBRID} algorithm on PIC-MAG on 6400 processors}
  \label{fig:hybrid-PICMAG-m6400}
\end{figure}

\begin{figure}[htb]
  \centering 
  \includegraphics[width=.8\linewidth]{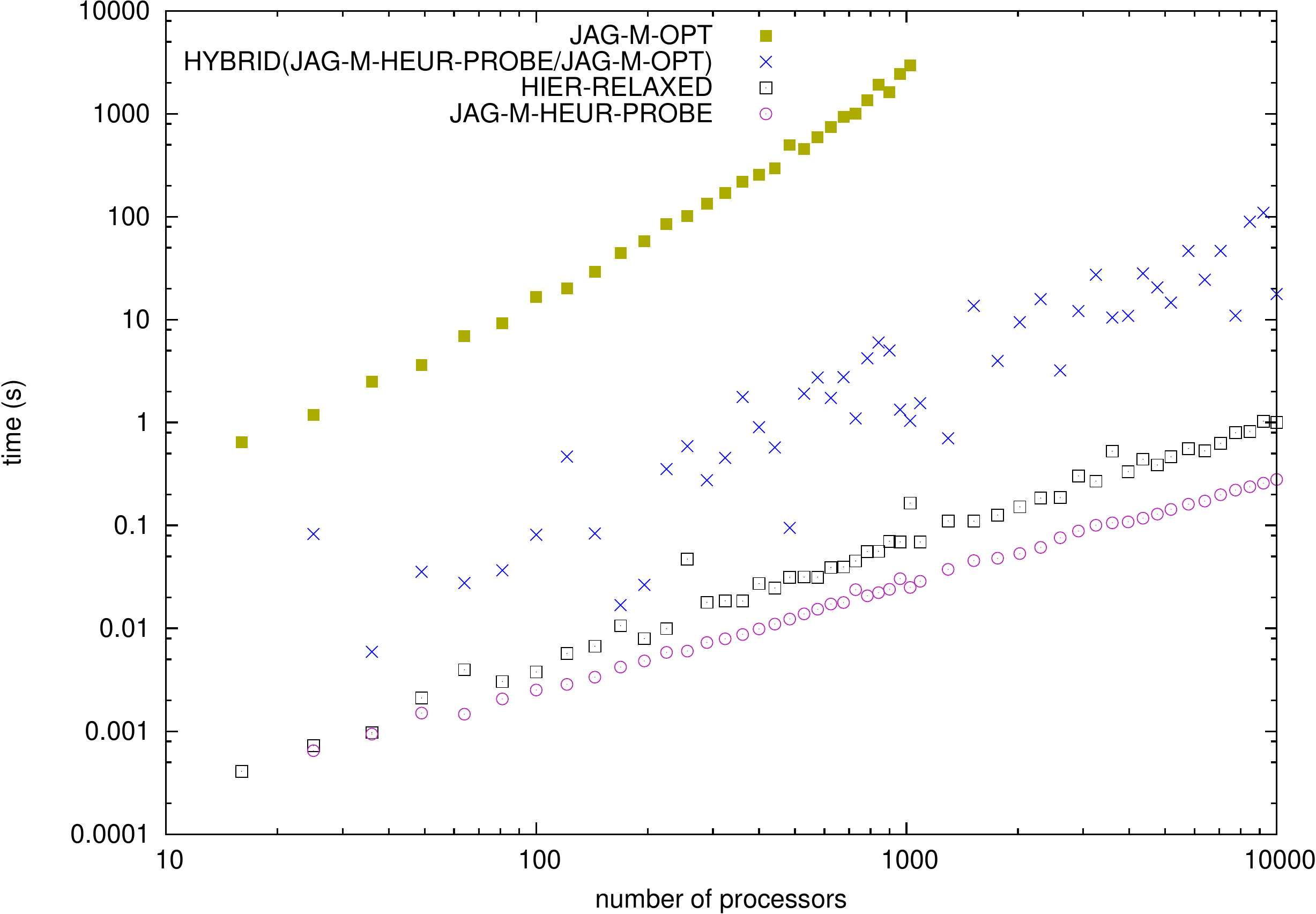}
  \vspace*{-2ex}
  \caption{Runtime of {\tt HYBRID} methods on PIC-MAG iter=10000}
  \label{fig:hybrid-runtime}
\end{figure}

The runtime of the algorithm is presented in
Figure~\ref{fig:hybrid-runtime}. It shows that the {\tt HYBRID}
algorithm is two or three orders of magnitude slower than the
heuristics but one to two orders of magnitude faster than {\tt
  JAG-M-OPT}. However, {\tt HYBRID} algorithms are likely to
parallelize pleasantly.

Some more engineering techniques could be applied to {\tt
HYBRID}. Different time/quality tradeoff could be obtained by stopping
the use of the {\em slow} algorithm in phase 2 when the improvement
become smaller than a given threshold. Using a 3-phase {\tt HYBRID}
mechanism could be another way of obtaining different trade-offs. 

{\tt HYBRID} is not the only option for algorithm engineering. One idea
that might lead to interesting time/quality tradeoff would be to avoid 
running dynamic programming algorithms all the way through. 
Early termination can be decided based on a time allocation or a targeted
maximum load.

Finally, different kind of iterative improvement algorithms could be
designed. For instance, on \nbproc-way jagged partition, {\tt
  JAG-M-PROBE} provides the optimal number of processors to use in
each stripe provided the partition in the main dimension, and {\tt
  JAG-M-ALLOC} provides the optimal partition in the main dimension
provided the number of processors allocated to each stripe. Applying
{\tt JAG-M-PROBE} and {\tt JAG-M-ALLOC} the one after the other as
long as the solution improves would be one interesting iterative algorithm.

\section{Conclusion}
\label{sec:ccl}

Partitioning spatially localized computations evenly among processors
is a key step in obtaining good performance in a large class of
parallel applications. In this work, we focused on partitioning a
matrix of non-negative integers using rectangular partitions to obtain
a good load balance. We introduced the new class of solutions called
\nbproc-way jagged partitions, designed polynomial optimal algorithms
and heuristics for \nbproc-way partitions. Using theoretical worst
case performance analyses and simulations based on logs of two real
applications and synthetic data, we showed that the {\tt
  JAG-M-HEUR-PROBE} and {\tt HIER-RELAXED} heuristics we proposed get
significantly better load balances than existing algorithms while
running in less than a second. We showed how {\tt HYBRID} algorithms
can be engineered to achieve better load balance but use significantly
more computing time. Finally, if computing time is not really a
limitation, one can use more complex algorithm such that {\tt
  JAG-M-OPT}.

Showing that the optimal solution for \nbproc-way jagged
partitions, hierarchical bipartitions and hierarchical $k$-partitions
with constant $k$ can be computed in polynomial time is a strong
theoretical result. However, the runtime complexity of the proposed
dynamic programming algorithm remains high. Reducing the polynomial
order of these algorithms will certainly be of practical interest.

We only considered computations located in a two dimensional field but
some applications, such as PIC-MAG and SLAC, might expose three or
more dimensions. A simple way of dealing with higher dimension would
be to project the space in two dimensions and using a two dimensional
partitioning algorithm, as we have done in some of the
applications. But this choice is likely to be suboptimal since it
drastically restrict the set of possible allocations. An alternative
would be to extend the classes of partitions and algorithm to higher
dimension. For instance, a jagged partitioning algorithm would
partition the space along one dimension and perform a projection to
obtain planes which will be partitioned in stripes and projected to
one dimensional arrays partitioned in intervals. All the presented
algorithms extend in more than two dimensions, therefore the problems
will stay in the same complexity class. However, the guaranteed
approximation is likely to worsen, the time complexity is likely to
increase (especially for dynamic programming based algorithms). Memory
occupation is also likely to become an issue and providing cache
efficient algorithm should be investigated. However, the increase of
the size of the solution space will provide better load balance than
partitioning the two dimensional projection.

We are also planning to integrate the proposed algorithms in a
distributed particle in cell simulation code. To optimize the
application performance, we will need to take into account
communication into account while partitioning the task. Dynamic
application will require rebalancing and the partitioning algorithm
should take into account data migration cost. Finally, to keep the
rebalancing time as low as possible, it might useful not to gather the
load information on one machine but to perform the repartitioning
using a distributed algorithm.

\section*{Acknowledgment}

We thank to Y. Omelchenko and H. Karimabadi for providing us with the
PIC-MAG data; and R. Lee, M. Shephard, and X. Luo for the SLAC
data.

\clearpage

\bibliographystyle{abbrv}
\bibliography{techreport}

\end{document}